\documentclass[10pt]{article}

\usepackage{graphicx} 
\usepackage{xspace}
\usepackage{ragged2e}
\usepackage{url}
\usepackage{mathtools}
\usepackage{mathrsfs}
\usepackage{amssymb}
\usepackage{amsthm}
\usepackage{bbm}
\usepackage{empheq}
\usepackage{latexsym}
\usepackage{enumitem}
\usepackage{eurosym}
\usepackage{dsfont}
\usepackage{appendix}
\usepackage{color} 
\usepackage[unicode]{hyperref}
\usepackage{frcursive}
\usepackage[utf8]{inputenc}
\usepackage[T1]{fontenc}
\usepackage{geometry}
\usepackage{multirow}
\usepackage{todonotes}
\usepackage{lmodern}
\usepackage{anyfontsize}
\usepackage{stmaryrd}
\usepackage{natbib}
\usepackage{cleveref}
\usepackage[english]{babel}
\usepackage[english=british]{csquotes}
\usepackage{float}
\usepackage{bm}
\usepackage{subfigure}
\everymath{\displaystyle}

\setcitestyle{numbers,open={[},close={]}}

\definecolor{red}{rgb}{0.7,0.15,0.15}
\definecolor{green}{rgb}{0,0.5,0}
\definecolor{blue}{rgb}{0,0,0.7}
\hypersetup{colorlinks, linkcolor={red},citecolor={green}, urlcolor={blue}}
			
\makeatletter \@addtoreset{equation}{section}

\newtheorem{theorem}{Theorem}[section]

\newtheorem{lemma}[theorem]{Lemma}

\newtheorem{definition}[theorem]{Definition}
\newtheorem{remark}[theorem]{Remark}
\newtheorem{result}{Result}

\def \E{\mathbb{E}}

\def\d{\mathrm{d}}

\title{Optimal hedging of an informed broker facing many traders}
\author{Philippe Bergault\thanks{CEREMADE, UMR CNRS 7534, Universit\'e Paris Dauphine-PSL,
Place de Lattre de Tassigny, 75775 Paris Cedex 16, France. e-mail: bergault@ceremade.dauphine.fr} \and 
Pierre Cardaliaguet\thanks{CEREMADE, UMR CNRS 7534, Universit\'e Paris Dauphine-PSL,
Place de Lattre de Tassigny, 75775 Paris Cedex 16, France. e-mail: cardaliaguet@ceremade.dauphine.fr}\and 
Wenbin Yan\thanks{École Polytechnique, Rte de Saclay, 91120 Palaiseau, France. e-mail: wenbin.yan@polytechnique.edu}}

\begin{document}

\maketitle

\begin{abstract}
    This paper investigates the optimal hedging strategies of an informed broker interacting with multiple traders in a financial market. We develop a theoretical framework in which the broker, possessing exclusive information about the drift of the asset's price, engages with traders whose trading activities impact the market price. Using a mean-field game approach, we derive the equilibrium strategies for both the broker and the traders, illustrating the intricate dynamics of their interactions. The broker's optimal strategy involves a Stackelberg equilibrium, where the broker leads and the traders follow. Our analysis also addresses the mean field limit of finite-player models and shows the convergence to the mean-field solution as the number of traders becomes large.

\end{abstract}

\tableofcontents

\section{Introduction}

Liquidity provision is a cornerstone of financial markets, with a significant portion of this activity occurring in over-the-counter (OTC) markets where brokers play a pivotal role in managing order flow. A key challenge for brokers arises when they possess private information about the price process, specifically the drift, and must interact with a large number of uninformed traders. This creates a complex environment where the broker seeks to exploit his informational advantage while managing the risks associated with inventory and market impact. The core of the problem lies in the broker’s ability to externalize risk in the market or internalize it to optimize his position without fully revealing his informational edge to the uninformed traders.\\

The study of externalization-internalization strategies has gained increasing attention in financial markets research. Externalization refers to the act of hedging or offloading a position in the market, while internalization involves holding the risk internally in the hope of favorable price movements or the arrival of offsetting trades. Several recent studies have focused on brokers’ strategies for balancing these two approaches in various markets. For instance, \citeauthor{butz2019internalisation} \cite{butz2019internalisation} develop a model that computes the internalization horizon and cost for foreign exchange (FX) brokers. Other works, such as \citeauthor{barzykin2021algorithmic} \cite{barzykin2021algorithmic,barzykin2022market}, explore how brokers choose to externalize risk only when their inventory exceeds certain thresholds (see also \citeauthor{barzykin2022dealing} \cite{barzykin2022dealing, barzykin2024market}). Additionally, empirical surveys like \citeauthor{schrimpf2019fx} \cite{schrimpf2019fx} highlight the growing use of internalization in FX markets, though externalization remains critical for risk management.\\

A closely related branch of the literature examines how brokers unwind stochastic order flow. \citeauthor{cartea2020trading} \cite{cartea2020trading} and \citeauthor{cartea2022double} \cite{cartea2022double} focus on optimal liquidation strategies in markets where order flow is stochastic, while \citeauthor{muhle2023pre} \cite{muhle2023pre} studies pre-hedging strategies to improve trading outcomes. These works emphasize the importance of managing inventory and market impact, particularly when brokers face informed or uninformed counterparties. Some closely related works are \citeauthor{nutz2023unwinding} \cite{nutz2023unwinding} and \citeauthor{bergault2025hedge} \cite{bergault2025hedge}, particularly suited for the context of central risk books. In contrast to these works, our paper examines a Stackelberg game in which a broker, possessing private information about the price drift, faces a continuum of uninformed traders. The goal is to determine how the broker can strategically externalize order flow while withholding key information from the market. The choice of a Stackelberg game is motivated by several applications. First, although in theory the broker’s activity in an interbank market is not perfectly observable by traders, some participants can still detect recurring execution patterns and infer information about the broker’s strategy, particularly if the broker plays a major role in providing liquidity. Moreover, \citeauthor{schoneborn2009liquidation} \cite{schoneborn2009liquidation} show that in certain contexts (specifically, when temporary market impact outweighs permanent impact), an agent seeking to liquidate their inventory may have an incentive to disclose their intention so that competitors supply liquidity: in such a context, our paper provides the optimal way to reveal intentions without disclosing too much information to uninformed traders. Finally, an obvious and increasingly relevant application is that of Automated Market Makers (AMMs) in decentralized finance. Indeed, \citeauthor{milionis2022automated} \cite{milionis2022automated} demonstrate that, in order to limit losses for liquidity providers, an AMM benefits from hedging through an externalization strategy on a centralized exchange (or potentially another decentralized exchange). These protocols operate under publicly predefined rules that are accessible to all (via smart contracts): the fact that the strategy is observable by all agents is an inherent characteristic of the decentralized system. An AMM deciding to adjust its externalization strategy based on private information must therefore carefully consider the information it is about to reveal.\\

Information asymmetry has long been a central theme in algorithmic trading research. For instance, \citeauthor{muhlekarbe2017information} \cite{muhlekarbe2017information} and \citeauthor{cont2023fast} \cite{cont2023fast} examine how traders with short-term informational advantages can exploit them to profit at the expense of slower or less informed traders. Similarly, \citeauthor{herdegen2023liquidity} \cite{herdegen2023liquidity} analyze liquidity provision under adverse selection, a scenario where one party holds superior information. There is also a rich literature on insider trading that started with  \citeauthor{kyle1989informed} \cite{kyle1989informed, kyle1985continuous}, \citeauthor{back1992insider} \cite{back1992insider} -- more recent work include \citeauthor{campi2011dynamic} \cite{campi2011dynamic} and \citeauthor{ccetin2021pricing} \cite{ccetin2021pricing}. A particularly relevant study by \citeauthor{cartea2022brokers} \cite{cartea2022brokers} examines a Stackelberg game where an informed broker faces both informed and uninformed traders, providing a useful framework for understanding how brokers can optimize trading strategies in the presence of information asymmetry. \citeauthor{bergault2024mean} \cite{bergault2024mean} extends this framework to consider a broker interacting with a large number of informed traders, each of whom may have different beliefs on the drift of the price process (see \citeauthor{baldacci2023mean} \cite{baldacci2023mean} for a related framework studying the optimal quotes of a broker facing a large number of informed traders).\\

Our work is grounded in the theory of mean field games (MFG), initially developed in the seminal contributions of Lasry and Lions \cite{lasry2006jeux, lasry2006jeuxii, lasry2007mean}, as well as Caines, Huang, and Malhamé \cite{huang2006large, huang2007large, huang2007invariance}. Specifically, we focus on MFG with a major player, a framework introduced by \citeauthor{huang2010large} \cite{huang2010large} and further explored in \citeauthor{bensoussan2016mean} \cite{bensoussan2016mean}, \citeauthor{carmona2017alternative} \cite{carmona2017alternative}, \citeauthor{carmona2016probabilistic} \cite{carmona2016probabilistic}, and \citeauthor{nourian2013eps} \cite{nourian2013eps}. Other relevant studies include \citeauthor{bertucci2020strategic} \cite{bertucci2020strategic}, and \citeauthor{cardaliaguet2020remarks} \cite{cardaliaguet2020remarks}. This theory has been applied to optimal trading problem for instance in \citeauthor{cardaliaguet2018mean} \cite{cardaliaguet2018mean} and \citeauthor{carmona2015probabilistic} \cite{carmona2015probabilistic}. In our model, similar to \citeauthor{bensoussan2016mean} \cite{bensoussan2016mean}, \citeauthor{guo2022optimization} \cite{guo2022optimization}, \citeauthor{moon2018linear} \cite{moon2018linear}, and \citeauthor{nguyen2012mean} \cite{nguyen2012mean}, we analyze a Stackelberg equilibrium where the major player acts as the leader. This approach is closely related to principal-agent problems involving a large number of agents, as explored in \citeauthor{carmona2017alternative} \cite{carmona2017alternative}, \citeauthor{elie2019tale} \cite{elie2019tale}, \citeauthor{elie2019contracting} \cite{elie2019contracting}, \citeauthor{elie2021mean} \cite{elie2021mean}, \citeauthor{hubert2023continuous} \cite{hubert2023continuous}, and \citeauthor{nutz2019mean} \cite{nutz2019mean}. Of particular relevance is the recent work by \citeauthor{djete2023stackelberg} \cite{djete2023stackelberg}, which bridges the gap between Stackelberg equilibria in games with a finite number of players and mean field games involving a major player.\\

There is still limited understanding of MFG problems involving information asymmetry. \citeauthor{sen2016mean} \cite{sen2016mean} and \citeauthor{firoozi2020epsilon} \cite{firoozi2020epsilon} explore scenarios in which small players only partially observe the actions of a major player, using nonlinear filtering techniques to reconstruct a fully observed system. In contrast, \citeauthor{casgrain2020mean} \cite{casgrain2020mean} examines a setting where agents hold differing beliefs about the underlying model. \citeauthor{bertucci2022mean} \cite{bertucci2022mean} studies a problem closer to the one considered here, addressing MFGs with shared but incomplete information among players and developing an associated master equation, although the players are symmetric in their roles. More recently, \citeauthor{becherer2023mean} \cite{becherer2023mean} investigates situations in which players must expend effort to acquire information about their own state. In a related line of research, \citeauthor{bergault2024mfg} \cite{bergault2024mfg} proposes a framework for Stackelberg games with numerous small agents, where the major player holds private information on a random variable associated with the index.\\


This paper builds on the framework introduced by \citeauthor{bergault2024mfg} \cite{bergault2024mfg}, contributing to the literature by presenting a game-theoretic model in which an informed broker interacts with a continuum of uninformed traders, modeled as a Stackelberg game. The broker, acting as the leader, makes strategic externalization decisions while concealing their informational advantage. The uninformed traders, as followers, respond to the broker’s trades without knowledge of the true drift.\\

Unlike \citeauthor{bergault2024mfg} \cite{bergault2024mfg}, where all players engage in an open-loop game and a relaxed problem for the broker is considered, we propose an open-loop formulation for the broker and a closed-loop formulation for the traders. We solve for the broker’s optimal strategy, deriving explicit results on how the broker balances externalization and inventory management in this asymmetric information setting. In particular, we establish the existence of a deterministic critical time at which the broker should fully disclose their private information, while ensuring it remains completely hidden until that moment. These results highlight the impact of private information on market dynamics, particularly with respect to liquidity provision and consumption.\\

While the most ideal model would involve a closed-loop formulation for all players, a detailed analysis of such a problem, incorporating the master equation and addressing a Nash equilibrium for asymmetric players with incomplete information, remains challenging and poorly understood. Advancing our understanding in this direction could be highly valuable for future research.\\

The remainder of this paper is structured as follows. Section \ref{model_setting_and_main_conclusion} provides a formal description of the model and heuristically states the main result for readers primarily interested in the outcomes. Section \ref{section_comments_limitations} formulates the Stackelberg game with an infinite number of traders, while Section \ref{section_solution} derives the MFG equilibrium for the uninformed traders and the broker's optimal externalization strategy. Finally, in Section \ref{section_comments_finite}, we demonstrate that the broker's optimal strategy in the mean-field setting is approximately optimal in the finite \( N \)-trader game.

\section{Mathematical Model and Main Results}\label{model_setting_and_main_conclusion}

In this section, we aim to establish the mathematical model of an informed broker interacting with multiple traders in a financial market. The discussion will focus on constructing the mathematical framework for the corresponding financial problem, which involves determining the optimal hedging strategies for an informed broker interacting with multiple uninformed traders. The discussion will not be fully rigorous at this stage. A rigorous formalization of the problem, both in the mean-field setting and in the finite case, will be presented in sections \ref{section_comments_limitations} and \ref{section_comments_finite}, respectively.

\subsection{Market Setup and Dynamics}\label{subsec_Model}

We consider a trading horizon \( T > 0 \) and a probability space \( \big(\Omega, \mathcal{F}, \mathbb{P}\big) \), under which all stochastic processes are defined (see Remark \ref{enlarge the space}). Let \( N \in \mathbb{N}^\star \) denote the number of traders in a market consisting of a single asset whose price process is denoted by \( S \). We denote by \( \nu_t^n \) and \( \nu_t^B \), respectively, the execution rates of trader \( n \) and the broker on the stock. The price process of the stock \( (S_t)_{t \in [0,T]} \) evolves according to the following dynamics:
\[
\d S_t = \left(b\,\bar{\nu}_t + \mu \right) \d t + \sigma \d W_t,
\]
where \( \sigma > 0 \) is a constant, \( \mu \) is a real-valued random variable with finite second moment, and \( \bar{\nu}_t := {\scriptstyle \frac{1}{N} \sum\nolimits_{n=1}^N }\nu_t^n \) represents the average execution rate of the traders. In this setup, the traders have a linear permanent impact on the price, which is observable by all market participants.\\

The inventory and cash processes of trader \( n \) are denoted by \( (Q^n_t)_{t \in [0,T]} \) and \( (X^n_t)_{t \in [0,T]} \), respectively. Trader \( n \) interacts with the broker at rate \( (\nu_t^n)_{t \in [0,T]} \), where \( \nu^n_t > 0 \) corresponds to buying and \( \nu^n_t < 0 \) corresponds to selling. The inventory dynamics and cash dynamics are given by:
\[
\d Q^n_t = \nu^n_t \d t,\quad
\d X^n_t = -\nu^n_t \big(S_t + \eta \nu^n_t \big) \d t,
\]
where \( \eta > 0 \) is the transaction cost charged by the broker (identical for all traders). The initial inventories \( (Q_0^n)_{1 \leq n \leq N} \), the random variable \( \mu \), and the Brownian motion \( W \) are assumed to be mutually independent.\\

We denote by \( (N Q^B_t)_{t \in [0,T]} \) and \( (N X^B_t)_{t \in [0,T]} \) the inventory and cash processes of the broker, respectively. The scaling factor $N$ is introduced to facilitate the mean-field limit analysis later. This adjustment is particularly useful when introducing the mean-field limit. The broker, who executes the trades of the traders, can also externalize trades in a lit market at rate \( (N \nu^B_t)_{t \in [0,T]} \). The broker's inventory and cash process evolve according to the dynamics:
\[
\d Q^B_t = \left(\nu^B_t - \bar \nu_t \right) \d t, \quad
\d X^B_t = \frac{1}{N} \sum_{n=1}^N \nu^n_t \big(S_t + \eta \nu^n_t \big) \d t - \nu^B_t \big(S_t + \eta^B \nu^B_t \big) \d t,
\]
where \( \eta^B > 0 \) is the execution cost associated with trading in the lit market.

\begin{remark}\label{enlarge the space}
Specifically, we would consider the problem on the probability space 
$\left( \Omega \times \mathbb{R}, \mathcal{F} \otimes \mathcal{B}(\mathbb{R}), \mathbb{P} \otimes \mathbb{Q} \right)$, where $(\Omega, \mathcal{F})$ is rich enough to support the independent random variables $\mu$, \( (Q_0^n)_{1 \leq n \leq N} \), and \( W \), while
the probability measure \( \mathbb{Q} \) could also be part of the broker's control. This would allow the broker to enrich the value of \( \mu \) without changing the problem to be solved, especially when $\mu$ is a discrete random variable. Thus, to simplify our notation and without loss of generality, we assume that \( \sigma(\mu) \) is already rich enough to ensure the problem is non-trivial.
\end{remark}

\subsection{Optimization Problems}\label{subsec_Model_Players}


Knowing the responses of the broker and other traders, the $n$-th trader aims to maximize the following objective function:
\begin{align*}
\mathbb{E} \left[ X^n_T + Q^n_T S_T - a\,\left( Q^n_T \right)^2 - \phi \int_0^T \left( Q^n_t \right)^2 \d t \right],
\end{align*}
over their set of admissible controls $(\nu^n_t)_{t \in [0,T]}$, where  $a,\phi > 0$ correspond to the risk aversion of the informed trader, and are identical for all traders.
By Ito's formula, it is easy to see that this is equivalent to maximizing
\begin{align}\label{model_setting_and_main_conclusion:tradersproblem}
\mathbb{E} \left[ \int_0^T \left\{ Q^n_t \left(b\,\bar \nu_t + \mu \right) - \eta \left( \nu^n_t \right)^2 - 2\,a Q^n_t \nu^n_t - \phi \left( Q^n_t \right)^2 \right\} \d t \right].
\end{align}

Knowing the responses of the traders, the broker seeks to maximize the following objective function:
\begin{align*}
\mathbb{E} \left[ X^B_T + Q^B_T S_T - a^B\left(Q^B_T\right)^2 - \phi^B \int_0^T \left( Q^B_t \right)^2 \d t \right],
\end{align*}
over his set of admissible controls $(\nu^B_t)_{t \in [0,T]}$, where $a^B, \phi^B > 0$ correspond to the broker's risk aversion.
Again, by Ito's formula, it is easy to see that this is equivalent to maximizing
\begin{align}\label{model_setting_and_main_conclusion:brokersproblem}
\mathbb{E} \left[ \int_0^T \left\{ Q^B_t \left(b \bar \nu_t + \mu \right) + \eta \frac{1}{N} \sum_{n=1}^N \left(\nu^n_t \right)^2 - \eta^B \left( \nu^B_t \right)^2 - 2\,a^B Q^B_t \left( \nu^B_t - \bar \nu_t \right) - \phi^B \left( Q^B_t \right)^2 \right\} \d t \right].
\end{align}

\begin{remark}\label{drift_for_traders}
    At time 0, the drift \( \mu \) is revealed only to the broker. Consequently, traders can infer the drift term solely based on the information leaked through the broker's execution rate. Mathematically, this is expressed as \( \mu_t := \mathbb{E} [ \mu | \sigma(\nu^B_s \, ; ~ 0 \leq s \leq t) ] \), which forms a martingale.
\end{remark}

\subsection{Main Results}

To obtain an approximate solution for the finite game, we propose a mean field game approach. Note that the propagation of chaos suggests that, in the limit, the execution rates of individual traders become independent, conditional on \( \nu^B \). In this framework, the broker no longer interacts with \( N \) individual traders but rather the distribution infinitely many traders.\\

Our main results can be stated informally as follows. We have identified the broker's optimal information-leaking strategy and the corresponding optimal control in the mean field setting. Furthermore, we have proven that this strategy remains approximately optimal in the finite case, as detailed in Theorem \ref{best_control_for_broker}, Remark \ref{rmk on b}, and Theorem \ref{thm.finite_1}. Here we give the informal version of them.

\begin{result}\label{Main_Result}
Assume that the initial inventories of the traders are i.i.d. random variables with finite fourth moments, and that the linear permanent impact of the traders' execution rates on the stock price \( b \) is sufficiently small.\\\vspace{-1.8ex}

In the case of infinitely many traders, there exists a critical time \( t_c \in [0,T] \) such that the broker's optimal information-leaking strategy is that no information about the drift term \( \mu \) is revealed by the broker before the critical time, and all information is disclosed at the critical time. The broker's optimal control is given by an explicit formula, piecewise analytic on \( [0,t_c] \) and \( (t_c, T] \).\\\vspace{-1.8ex}

Moreover, in the case of finitely many traders, there exists \( N_0 \in \mathbb{N} \) such that for all \( N \geq N_0 \), the optimal control in the infinite case is a \( C/\sqrt{N} \)-optimal one for the broker, where \( C \) is a constant independent of \( N \).
\end{result}

\section{Mean Field Game and Stackelberg Problem}\label{section_comments_limitations}

In this section, we formalize the Stackelberg problem between the broker and the representative trader in the mean field setting. The original game, which involves one broker and many traders, simplifies to a game between the broker and a single representative trader, thereby significantly reducing the problem's complexity. We will present the available information for both the representative trader and the broker, define their respective optimization problems, and outline the game they participate in, to properly introduce the corresponding Stackelberg framework.

\subsection{MFG Equilibrium of the Traders}\label{section_mfg_equilibrium_for_traders}
In this section, we suppose that the execution rate of the broker, \( (\nu^B_t)_{t \in [0,T]} \), is fixed. We first consider the admissible strategies for the representative trader, which are defined as follows:
\[
\mathcal{A}^r = \left\{ \nu = (\nu_t)_{t \in [0,T]} \left| \nu_t \text{ is } \mathcal{F}^{r}_t \text{-progressively measurable, and } \mathbb{E} \left[ \int_0^T \nu_t^2 \d t \right] < +\infty \right. \right\},
\]
where \( \mathcal{F}^{r} \) denotes the filtration generated by the initial condition of the representative trader's inventory process \( Q_0 \) and the broker's execution rate \( (\nu^B_t)_{t \in [0,T]} \).\\

To introduce the mean field game of traders, given the broker’s execution rate \( \nu^B \), for \( \nu \in \mathcal{A}^r \), we consider the associated inventory process \( (Q_t)_{t \in [0,T]} \) of the representative trader as
\[
Q_t = Q_0 + \int_0^t \nu_u \, \d u,
\]
and the mean field execution rate \( (\bar \nu_t)_{t \in [0,T]} \) of the traders as
\[
\bar \nu_t = \int_{\mathbb{R}} x\, m_t(\d x),
\]
where the process \( (m_t)_{t \in [0,T]} \) takes values in \( \mathcal{P}_2 (\mathbb{R}) \), representing the distribution of the other traders' execution rates at time \( t \), conditional on \( \mathcal{F}_t := \sigma(\{\nu^B_s ; 0 \leq s \leq t\}) \).\\

Hence, given the broker's execution rate \( \nu^B \), the representative trader should try to maximize the objective function as follow, which is the mean field version of \eqref{model_setting_and_main_conclusion:tradersproblem},
\begin{align}\label{trader_obj_infinite}
H[\nu^B](\nu) = \mathbb{E} \left[ \int_0^T \left\{ Q_t \left(b\,\bar \nu_t + \mu \right) - \eta \left( \nu_t \right)^2 - 2\,a Q_t \nu_t - \phi \left( Q_t \right)^2 \right\} \d t \right],
\end{align}
with \( b > 0 \) representing the market impact of the traders, \( \eta \) representing the transaction costs charged by the broker to the traders, and \( a, \phi > 0 \) corresponding to the risk aversion of the informed traders.\\

Finally, we can give the following definition,
\begin{definition}\label{minorsol}
Given an admissible control $\nu^B$ of the broker, an MFG equilibrium associated with $\nu^B$ is a pair $(\nu^{\nu^B}, m^{\nu^B})$ of processes where $\nu^{\nu^B} \in \mathcal{A}^r$ and $m^{\nu^B}$ is an $\mathcal{F}^{r}_t$-progressively measurable random process taking values in $\mathcal{P}_2(\mathbb{R})$, and 
\begin{itemize}
    \item[$(i)$] $H[\nu^B](\nu^{\nu^B}) = \underset{\nu \in \mathcal{A}^r}{\sup} H[\nu^B](\nu);$
    \item[$(ii)$] $m^{\nu^B}_t$ is the distribution of $\nu^{\nu^B}_t$ conditional on $\mathcal{F}_t = \sigma \left( \{ \nu^B_s; s \leq t \} \right)$ for Lebesgue-almost every $t \in [0,T]$.
\end{itemize}
\end{definition}

\subsection{Stackelberg Equilibrium of the Informed Broker}\label{section_stackelberg_equilibrium_for_broker}

The set of admissible strategies for the informed broker is defined as:
\[
\begin{aligned}
\mathcal{A}^B = \left\{ \nu^B = (\nu^B_t)_{t \in [0,T]} \ \middle| \ 
\begin{array}{l}
\nu^B_t \text{ is independent of } Q_0 \text{ and } 
\mathbb{E} \left[ \int_0^T (\nu^B_t)^2 \, \d t \right] < +\infty.
\end{array}
\right\}.
\end{aligned}
\]
and we write \( \mathcal{F}^{B} \) for the filtration generated by the \( \nu^B \), chosen by the broker in the mean field game. Note that if \( \mu \) is a constant, the problem is trivial and not considered here. If \( \mu \) takes only discrete values, as in \citeauthor{bergault2024mfg} \cite{bergault2024mfg}, the broker would then be allowed to enlarge the probability space as we mentioned in Remark \ref{enlarge the space}. Although the independence of the \( \nu^B \) with \( Q_0 \) may seem strange at the beginning, it is saying nothing but that the broker is playing an open-loop game, since he knows nothing about the conditions of the traders throughout the time.\\

To introduce the Stackelberg game of the broker, given the representative trader's execution rate \( \nu \), for $\nu^B \in \mathcal{A}^B$, we define the associated inventory process of the informed broker as
\[
Q^B_t = Q^B_0 + \int_0^t \left( \nu^B_u - \bar \nu_u \right) \d u.
\]

Hence, given the distribution \( (m_t)_{t \in [0,T]} \) of the traders, taking the mean field version of problem \eqref{model_setting_and_main_conclusion:brokersproblem}, the broker would try to optimize the following objective function,
\begin{equation}\label{broker_obj_infinite}
\scalebox{0.9}{$
\begin{aligned}
    H^{B}[m](\nu^B) = \mathbb{E} \left[ \int_0^T \left\{ Q^B_t \left(b\,\bar \nu_t + \mu \right) + \eta \int_{\mathbb{R}} x^2 m_t(\d x) - \eta^B \left( \nu^B_t \right)^2 - 2\,a^B Q^B_t \left( \nu^B_t - \int_\mathbb{R} x\, m_t(\d x) \right) - \phi^B \left( Q^B_t \right)^2 \right\} \d t \right].
\end{aligned}
$}
\end{equation}
with $\eta^B > 0$ representing the transaction costs in the lit market, and $a^B, \phi^B > 0$ the risk aversion parameters for the broker.\\

Finally, we can give the following definition,
\begin{definition}\label{majorsol}
The problem of the broker consists of solving
$$\underset{\nu^B \in \mathcal{A}^B}{\sup} \quad \underset{(\nu^{\nu^B}, m^{\nu^B})}{\inf} H^{B}[m^{\nu^B}] (\nu^B),$$
where the infimum is taken over all MFG equilibria $(\nu^{\nu^B}, m^{\nu^B})$ associated with $\nu^B$.
\end{definition}
We will find in the next section that in such an MFG the equilibrium is in fact unique, so that no infimum is actually required in the definition.

\section{Optimal Information Leaking Strategy}\label{section_solution}

In this part, we aim to determine the optimal strategies for both the representative trader and the broker. We achieve this by solving the backward stochastic Hamilton-Jacobi equations associated with the optimization problems in Definition \ref{minorsol} and Definition \ref{majorsol}, which are both linear-quadratic optimal stochastic control problems with random coefficients.\\

To this end, we will introduce a specific form of the Hamilton-Jacobi equation and find its unique solution in Subsection \ref{general_HJ_equation}: this part is largely inspired by \citeauthor{bismut1976linear} \cite{bismut1976linear}. 
With this result in mind, and by applying the maximum principle, we will derive the unique optimal control for the representative trader in problem \eqref{minorsol} (Lemma \ref{lem.hatnu}), given the broker's execution rate process \( \nu^B \), in Subsection \ref{Optimal_Control_Representative_Trader_Mean_Field_Setting}. 
Subsequently, we will identify the optimal information-leaking strategy and the broker's associated unique optimal control (Theorem \ref{best_control_for_broker}) in Subsection \ref{Optimal_Control_Broker_Mean_Field_Setting}.

\subsection{Explicit Computation for a Solution to a Backward HJ Equation}\label{general_HJ_equation}

We investigate the solution to an HJ equation of the form 
\begin{equation}\label{HJexp}
\begin{array}{l}
du_t = H_t(x,Du_t)\, dt + dM_t(x)\qquad \text{in } (0,T)\times \mathbb{R}, \\
u_T(x) = g_T(x) \qquad \text{in } \mathbb{R},
\end{array}
\end{equation}
where 
\begin{equation}\label{notgenH}
H_t(x,\xi) = A_t + B_t x + \frac{C_t x^2}{2} + D_t \xi + E_t x\xi + \frac{F_t \xi^2}{2},
\end{equation}
and 
$$
g_T(x) = G_1 + G_2x + \frac{G_3x^2}{2},
$$
with $(A_t,B_t,C_t,D_t,E_t,F_t)$ adapted to a filtration $(\mathcal{F}_t)$ (satisfying the usual assumptions) and bounded, and $(G_1, G_2, G_3)$ being $\mathcal{F}_T$-measurable and bounded. According to previous research (see \citeauthor{bismut1976linear} \cite{bismut1976linear}), we know that there exists a unique solution for such an equation. We look for $(u, M)$ in the form 
$$
u_t(x) = \alpha_t + \beta_t x + \frac{\gamma_t x^2}{2},
$$
and 
$$
M_t(x) = Z_{\alpha,t} + Z_{\beta,t} x + \frac{Z_{\gamma,t} x^2}{2},
$$
where $(\alpha_t, \beta_t, \gamma_t)$ and $(Z_\alpha, Z_\beta, Z_\gamma)$ are adapted to $(\mathcal{F}_t)$, and $(Z_\alpha, Z_\beta, Z_\gamma)$ form an $(\mathcal{F}_t)$-martingale. In view of the equation, $(\alpha_t, \beta_t, \gamma_t)$ and $(Z_\alpha, Z_\beta, Z_\gamma)$ must satisfy: 
\begin{align*}
d\alpha_t + x d\beta_t + \frac{x^2}{2} d\gamma_t = \left( A_t + B_t x + \frac{C_t x^2}{2} + D_t(\beta_t + \gamma_t x) + E_t x(\beta_t + \gamma_t x) + \frac{F_t (\beta_t + \gamma_t x)^2}{2} \right) dt \\
+ dZ_{\alpha,t} + x dZ_{\beta,t} + \frac{x^2}{2} dZ_{\gamma,t}.
\end{align*}
This yields the system 
$$
\left\{
\begin{array}{ll}
(i) & d\gamma_t = \left(C_t + 2E_t \gamma_t + F_t \gamma_t^2 \right) dt + dZ_{\gamma,t}, \qquad \gamma_T = G_3,\\ 
(ii) & d\beta_t = \left( B_t + D_t \gamma_t + E_t \beta_t + F_t \beta_t \gamma_t \right) dt + dZ_{\beta,t}, \qquad \beta_T = G_2,\\ 
(iii) & d\alpha_t = \left( A_t + D_t \beta_t + \frac{F_t \beta_t^2}{2} \right) dt + dZ_{\alpha,t}, \qquad \alpha_T = G_1.
\end{array}\right.
$$
We assume that $C_t$, $E_t$, $F_t$, and $G_3$ are deterministic and such that the Riccati equation 
$$
\frac{d}{dt} \gamma_t = C_t + 2E_t \gamma_t + F_t \gamma_t^2, \qquad  \gamma_T = G_3
$$
has a unique solution. Then $Z_\gamma = 0$. For $\beta_t$ we find 
$$
\beta_t = e^{-\int_t^T (E_s + F_s \gamma_s)ds} G_2 - \int_t^T e^{-\int_t^s (E_r + F_r \gamma_r)dr} (B_s + D_s \gamma_s) ds 
- \int_t^T e^{-\int_t^s (E_r + F_r \gamma_r)dr} dZ_{\beta,s}.
$$
Taking the conditional expectation with respect to $\mathcal{F}_t$, we obtain
$$
\beta_t = e^{-\int_t^T (E_s + F_s \gamma_s)ds} \mathbb{E}[G_2|\mathcal{F}_t] - \int_t^T e^{-\int_t^s (E_r + F_r \gamma_r)dr} \left( \mathbb{E}[B_s|\mathcal{F}_t] + \mathbb{E}[D_s|\mathcal{F}_t]\gamma_s \right) ds.
$$

\subsection{Optimal Control of the Representative Trader}\label{Optimal_Control_Representative_Trader_Mean_Field_Setting}

Throughout this subsection, we suppose that the filtration $\mathcal{F}_t$, representing the information disclosed by the informed broker, is fixed, and we discuss the optimal control of the representative trader in this case, that is, to consider,
$$
\max_{\nu_t} \mathbb{E} \left[ \int_0^T \left\{ Q_t \left(b\,\bar \nu_t + \mu_t \right) - \eta \left( \nu_t \right)^2 - 2\,a Q_t \nu_t - \phi \left( Q_t \right)^2 \right\} \d t \right],
$$
where 
$$
dQ_t = \nu_t \, dt,
$$
and where $\mu_t = \mathbb{E}[\mu\ |\mathcal{F}^r_t] = \mathbb{E}[\mu\ |\mathcal{F}_t]$, and $(\nu_t)$ is $(\mathcal{F}^r_t)$-progressively measurable. In this case, the value function $u$ of the representative trader satisfies an HJ equation of the form \eqref{HJexp}, with
\begin{align*}
H_t(q,\xi) &= \min_{\nu \in \mathbb{R}} \left\{ -\nu\xi - q \left(b\,\bar \nu_t + \mu_t \right) + \eta \nu^2 + 2\,a q \nu + \phi q^2 \right\} \\
& = -q \left(b\,\bar \nu_t + \mu_t \right) + \phi q^2 - \frac{1}{4\eta} |\xi - 2a q|^2,
\end{align*}
where the minimum is reached for $\tilde{\nu}$, given by 
$$
\tilde{\nu} = \tilde{\nu}(q, \xi) = \frac{1}{2\eta} (\xi - 2a q).
$$
With the notation of the previous subsection, we have
$$
A_t = 0, \; B_t = - \left(b\,\bar \nu_t + \mu_t \right), \; C_t = \phi - \frac{a^2}{\eta}, \; D_t = 0, \; E_t = \frac{a}{\eta}, \; F_t = -\frac{1}{2\eta}, \qquad g_T \equiv 0.
$$
Note that, aside from $B_t$, all the coefficients are deterministic. This leads to the system 
\begin{equation}\label{riccati_equation}
    \left\{
        \begin{array}{ll}
        (i) & d\gamma_t = \left(\phi - \frac{a^2}{\eta} + 2\frac{a}{\eta} \gamma_t - \frac{1}{2\eta} \gamma_t^2 \right) dt, \qquad \gamma_T = 0, \\ [1ex]
        (ii) & d\beta_t = \left(- \left(b\,\bar \nu_t + \mu_t \right) + \frac{2a - \gamma_t}{2\eta} \beta_t \gamma_t \right) dt + dZ_{\beta,t}, \qquad \beta_T = 0, \\ [1ex]
        (iii) & d\alpha_t = \left(-\frac{1}{4\eta} \beta_t^2 \right) dt + dZ_{\alpha,t}, \qquad \alpha_T = 0.
        \end{array}\right.
\end{equation}
The first equation has a (deterministic) analytic solution defined on \( [0,T] \), as guaranteed by classical results on ordinary differential equations. On the other hand, 
\begin{align}\label{defbeta}
\beta_t = -\int_t^T \Gamma_{s,t} \mathbb{E} \left[ - \left(b\,\bar \nu_s + \mu_s \right) \Big| \mathcal{F}_t \right] ds \notag
= b \int_t^T \Gamma_{s,t} \mathbb{E}[\bar \nu_s |\mathcal{F}_t] ds + \mu_t \int_t^T \Gamma_{s,t} ds,
\end{align}
where
$$
    \Gamma_{s,t} = \exp \left( \int_s^t \frac{\gamma_r - 2a}{2\eta} \, dr \right).
$$
Hence, the optimal dynamics for the representative trader are 
\begin{equation}\label{defhatnu}
d\hat{Q}_t = \hat{\nu}_t dt = \frac{1}{2\eta} \left( (\gamma_t - 2a)\hat{Q}_t + \beta_t \right) dt, \qquad \hat{Q}_0 = Q_0,
\end{equation}
where $Q_0$ (the initial distribution of the traders' inventories) is a random variable independent of all the others. Thus, 
\begin{equation}\label{defhatQ}
\hat{Q}_t = \Gamma_{0,t} \cdot Q_0 + \int_0^t \frac{1}{2\eta} \Gamma_{s,t} \beta_s ds.
\end{equation}
Moreover, we know that at the equilibrium, $\bar{\nu}$ satisfies $\bar{\nu}_t = \mathbb{E}[\hat{\nu}_t | \mathcal{F}_t]$ for all $t \in [0,T]$. Therefore,
\begin{align*}
    \bar{\nu}_t &= \mathbb{E} \left[ \frac{1}{2\eta} \left( (\gamma_t - 2a) \hat{Q}_t + \beta_t \right) \Big| \mathcal{F}_t \right] \\
    &= \frac{\gamma_t - 2a}{2\eta} \Gamma_{0,t} \cdot \bar{Q}_0 
     + b \frac{\gamma_t - 2a}{4\eta^2} \int_0^t \Gamma_{s,t} \int_s^T \Gamma_{r,s} \mathbb{E}[\bar{\nu}_r |\mathcal{F}_s] dr \, ds \\
     & \quad + \frac{\gamma_t - 2a}{4\eta^2} \int_0^t \Gamma_{s,t} \mu_s \int_s^T \Gamma_{r,s} dr \, ds \\
    & \quad + b \int_t^T \frac{1}{2\eta} \Gamma_{s,t} \mathbb{E}[\bar{\nu}_s |\mathcal{F}_t] ds + \mu_t \int_t^T \frac{1}{2\eta} \Gamma_{s,t} ds.
\end{align*}
On the normed space $\mathcal{D}$ of $\mathcal{F}_t$-progressively measurable processes $\xi$ with finite norm 
$$
\|\xi\|_{\mathcal{D}} := \mathbb{E} \left[ \sup_{0 \leq t \leq T} |\xi_t|^2 \right]^{\frac 12},
$$ 
we define the bounded linear operators $L, \tilde{L}: \mathcal{D} \to \mathcal{D}$ such that for any $\xi \in \mathcal{D}$ and $t \in [0,T]$,
\begin{align}\label{defLtildeL}
    L\xi_t &=  \int_0^t \int_s^T  A_{r,s,t} \mathbb{E}[\xi_r | \mathcal{F}_s] dr ds 
    + \int_t^T B_{s,t} \mathbb{E}[\xi_s | \mathcal{F}_t] ds, \\
    \tilde{L}\xi_t &= \tilde{L}^{C,D}\xi_t = \int_0^t C_{s,t} \xi_s ds + \xi_t D_t,
\end{align}
where
\[
A_{r,s,t} = \frac {\gamma_t - 2a}{4\eta^2} \Gamma_{r,s} \Gamma_{s,t},~ 
B_{s,t} = \frac{1}{2\eta} \Gamma_{s,t},~ 
C_{s,t} = \frac{\gamma_t - 2a}{4\eta^2} \Gamma_{s,t} \int_s^T \Gamma_{r,s} dr,~ 
D_t = \frac{1}{2\eta} \int_t^T \Gamma_{s,t} ds.
\]
Assuming that $b$ is small, i.e., $0 < b < \|L\|_{B(\mathcal{D})}^{-1}$, the fixed point $\bar{\nu}$ exists and is unique in $\mathcal{D}$, and can be written as 
\begin{equation}\label{defbarnu}
    \bar{\nu}_t = (I - bL)^{-1}(z \cdot \bar{Q}_0 + \tilde{L}\mu_\cdot),
\end{equation}
where $z \in C^\omega$ is defined by 
$$
z_t = \frac{\gamma_t - 2a}{2\eta} \Gamma_{0,t}, \qquad t \in [0,T].
$$
Given a pair of continuous maps $(E,F) = (E_{s,t},F_t)$, we note that 
\begin{align*}
L \circ \tilde{L}^{E,F} \mu_t &=  \int_0^t \int_s^T  A_{r,s,t} \mathbb{E}\left[\int_0^r E_{u,r} \mu_u du + \mu_r F_r \Big| \mathcal{F}_s \right] dr ds 
+ \int_t^T B_{s,t} \mathbb{E}\left[\int_0^s E_{u,s} \mu_u du + \mu_s F_s \Big| \mathcal{F}_t \right] ds \\
&=  \int_0^t \int_s^T  A_{r,s,t} \left( \int_0^r E_{u,r} \mu_u du + \mu_s\left(\int_s^r E_{u,r} du + F_r \right)\right) dr ds \\
&\quad + \int_t^T B_{s,t} \left( \int_0^t E_{u,s} \mu_u du + \mu_t\left(\int_t^s E_{u,s} du + F_s \right)\right) ds \\
&= \int_0^t \mu_u \left( \int_u^t \int_s^T  A_{r,s,t} E_{u,r} dr ds\right) du 
  + \int_0^t \mu_s \left( \int_s^T \int_s^r A_{r,s,t} \left( E_{u,r} du + F_r \right) du dr \right) ds \\
&\quad + \int_0^t \mu_u \left( \int_t^T B_{s,t} E_{u,s} ds\right) du + \mu_t \left( \int_t^T B_{s,t}\left(\int_t^s E_{u,s} du + F_s \right) ds \right) \\
&= \tilde{L}^{\tilde{E},\tilde{F}} \mu_t,
\end{align*}
with 
\begin{equation}\label{defhatL}
\left\{
\begin{array}{ll}
\tilde{E}_{s,t} &= \int_s^t \int_u^T  A_{r,u,t} E_{s,r} dr du + \int_s^T A_{r,s,t}\left( \int_s^r E_{u,r} du + F_r \right) dr + \int_t^T B_{u,t} E_{s,u} du,\\
\tilde{F}_t &= \int_t^T B_{s,t} \left(\int_t^s E_{u,s} du + F_s \right) ds.    
\end{array}
\right.
\end{equation}
Let $\hat{L} : \mathcal{C} \to \mathcal{C}$ be the operator such that $(\tilde{E}, \tilde{F}) = \hat{L}[E,F] = (\hat{L}_1[E,F], \hat{L}_2[E,F])$, where $\mathcal{C}$ is a normed space of function pairs $(E,F)$ in $C^0([0,T]^2) \times C^0([0,T])$ with finite norm $\|E\|_{C^0} \vee \|F\|_{C^0}$. Restricting $b$ further if necessary, i.e., $0 < b < \min\{\|L\|_{B(\mathcal{D})}^{-1}, \|\hat{L}\|_{B(\mathcal{C})}^{-1}\}$, we find that 
\begin{align*}
\beta_s &= b  \int_s^T \Gamma_{r,s} \mathbb{E}[\bar{\nu}_r | \mathcal{F}_s] dr + \mu_s \int_s^T \Gamma_{r,s} dr \\
&= b  \int_s^T \Gamma_{r,s} (I - bL)^{-1} z_r dr \cdot \bar{Q}_0 + \mu_s \int_s^T \Gamma_{r,s} dr 
 + \sum_{k=0}^\infty b^{k+1} \int_s^T \Gamma_{r,s} \mathbb{E}\left[L^k \circ \tilde{L} \mu_r | \mathcal{F}_s\right] dr \\
&= b  \int_s^T \Gamma_{r,s} (I - bL)^{-1} z_r dr \cdot \bar{Q}_0 + \mu_s \int_s^T \Gamma_{r,s} dr 
 + \sum_{k=0}^\infty b^{k+1} \int_s^T \Gamma_{r,s} \mathbb{E}\left[ \int_0^r \hat{L}_1^k[C,D]_{u,r} \mu_u du + \hat{L}_2^k[C,D]_r \mu_r \Big| \mathcal{F}_s \right] dr \\
&= \beta_s^1 \cdot \bar{Q}_0 + \int_0^s \beta^2_u \mu_u du + \beta_s^3 \mu_s,
\end{align*}
where
\begin{equation}\label{beta_parameter_definition}
\left\{
\begin{array}{ll}
    \beta_s^1 &= b  \int_s^T \Gamma_{r,s} (I - bL)^{-1} z_r dr, \\
    \beta_s^2 &= b  \int_s^T \Gamma_{r,s} \sum_{k=0}^\infty b^k \hat{L}_1^k [C,D]_{u,r} dr, \\
    \beta_s^3 &= \int_s^T \Gamma_{r,s} \left[ 1 + \sum_{k=0}^\infty b^{k+1} \left( \int_s^r \hat{L}_1^k[C,D]_{u,r} du + \hat{L}_2^k[C,D]_r \right) \right] dr.
\end{array}
\right.
\end{equation}

Hence, by combining the equation above with equations \eqref{defhatQ} and \eqref{defbarnu}, we have the expression for $\hat{\nu}$:
{\small
\begin{align*}
    \hat{\nu}_t &= \frac{\gamma_t - 2a}{2\eta} \left( \Gamma_{0,t} \cdot Q_0 + \int_0^t \frac{1}{2\eta} \Gamma_{s,t} \beta_s ds \right) + \frac{1}{2\eta} \beta_t \\
    &= \frac{\gamma_t - 2a}{2\eta} \left( \Gamma_{0,t} \cdot Q_0 
    + \bar{Q}_0 \cdot \int_0^t \frac{1}{2\eta} \Gamma_{s,t} \beta_s^1 ds 
    + \int_0^t \frac{1}{2\eta} \Gamma_{s,t} \int_0^s \beta_u^2 \mu_u du ds 
    + \int_0^t \frac{1}{2\eta} \Gamma_{s,t} \beta_s^3 \mu_s ds \right) \\
    &\quad + \frac{1}{2\eta} \left( \bar{Q}_0 \beta^1_t 
    + \int_0^t \beta_s^2 \mu_s ds + \beta_t^3 \mu_t \right) \\
    &= \bar{A}_t \bar{Q}_0 + \bar{B}_t Q_0 + \int_0^t \bar{C}_{s,t} \mu_s ds + \bar{D}_t \mu_t,
\end{align*}
}
where
\begin{equation}\label{barA_barB_barC_barD_definition}
    \left\{
    \begin{aligned}
        \bar{A}_t &= \frac{\beta_t^1}{2\eta} + \frac{\gamma_t - 2a}{4 \eta^2} \int_0^t \Gamma_{s,t} \beta_s^1 ds, \\
        \bar{B}_t &= \frac{\gamma_t - 2a}{ 2 \eta } \Gamma_{0,t}, \\
        \bar{C}_{s,t} &= \frac{\gamma_t - 2a}{ 4 \eta^2 } \left( \beta_s^2 \int_s^t \Gamma_{r,t} dr + \Gamma_{s,t} \beta_s^3 \right) + \frac{\beta_s^2}{2\eta}, \\
        \bar{D}_t &= \frac{\beta_t^3}{2\eta}.
    \end{aligned}
    \right.
\end{equation}

\begin{lemma}\label{lem.hatnu}
Assuming $b > 0$ is small enough, there exists $(\bar{A}, \bar{B}, \bar{C}, \bar{D}) \in C^{\omega}([0,T] \times [0,T], \mathbb{R}^4)$, as defined in \eqref{barA_barB_barC_barD_definition}, such that for any martingale $(\mu_t)$, the mean field control $\hat{\nu} = \hat{\nu}(Q_0,\mu_\cdot)$ is given by
\begin{equation}\label{hatnu_solution}
\hat{\nu}_t = \hat{\nu}(\mu_\cdot)_t = \bar{A}_t \bar{Q}_0 + \bar{B}_t Q_0 + \int_0^t \bar{C}_{s,t} \mu_s ds + \mu_t \bar{D}_t, \qquad t \in [0,T].
\end{equation}
Hence, we also have the conditional expectation of the first and second moments of $\hat{\nu}$ with respect to $\mathcal{F}_t$:
\begin{equation}\label{first_second_conditional_moment_of_hatnu}
\left\{
\begin{aligned}
    \bar \nu_t &= \mathbb E [\hat \nu_t | \mathcal F_t] = ( \bar A_t +\bar B_t ) \bar Q_0 + \int_0^t \bar C_{s,t} \mu_s   ds+ \mu_t \bar D_t, \\
    \bar M_t   &:= \mathbb E [\hat \nu_t^2 | \mathcal F_t] = \left( \bar A_t^2 + 2 \bar A_t \bar B_t \right) \cdot \bar Q_0^2 + \bar B_t^2 \cdot \mathbb E[Q_0^2] \\
     & \qquad\qquad + 2 \left( \bar A_t + \bar B_t \right) \cdot \left( \int_0^t \bar C_{s,t} \mu_s   ds+ \mu_t \bar D_t \right) \cdot \bar Q_0 + \left( \int_0^t \bar C_{s,t} \mu_s   ds+ \mu_t \bar D_t \right)^2,
\end{aligned}
\right.\qquad t \in [0,T].
\end{equation}
\end{lemma}

\begin{remark}\label{rmk on b}
    By "$\,b>0$ small enough", we mean precisely that,
    \begin{equation}\label{bound of b}
        0 < b < \exp\left( - 2T\frac{a \vee \sqrt{\eta\phi}}{\eta}\right) \cdot \frac{\eta^2\,a^{-1} \wedge \eta^{\frac 32}\,\phi^{-\frac 12} \wedge \eta}{T^2+T}.
    \end{equation}
\end{remark}

\begin{proof}
    Recall the definition of $L$ in formula \eqref{defLtildeL}. 
    Note that by the explicit formula of $\gamma$, we have 
    $$\|A\|_{C^0([0,T]^3)} \leq \frac{a \vee \sqrt{\eta\phi}}{2\eta^2}\exp\left(2T\frac{a \vee \sqrt{\eta\phi}}{\eta}\right), \qquad\|B\|_{C^0([0,T]^2)}\leq \frac{1}{2\eta}\exp\left(2T\frac{a \vee \sqrt{\eta\phi}}{\eta}\right).$$ 
    Hence, by Doob's inequality, we have
    {\small
    \begin{align*}
        \|L\xi\|_{\mathcal D}^2 &= \mathbb{E} \left[ \sup_{0 \leq t \leq T} \left| L\xi_t \right|^2 \right] \\
        &\leq 2 \cdot \mathbb{E} \left[ \sup_{0 \leq t \leq T} \left| \int_0^t \int_s^T A_{r,s,t} \mathbb{E} \left[ \xi_r | \mathcal{F}_r \right] dr ds \right|^2 \right] 
         + 2 \cdot \mathbb{E} \left[ \sup_{0 \leq t \leq T} \left| \int_t^T B_{s,t} \mathbb{E} \left[ \xi_s | \mathcal{F}_t \right] ds \right|^2 \right] \\
        &\leq 2\|A\|_{C^0}^2 \cdot \mathbb{E} \left[ \sup_{0 \leq t \leq T} \left( \int_0^t \int_s^T \mathbb{E} \left[ \left| \xi_r \right| \Big| \mathcal{F}_s \right] dr ds \right)^2 \right] 
         + 2\|B\|_{C^0} \cdot \mathbb{E} \left[ \sup_{0 \leq t \leq T} \left( \int_t^T \mathbb{E} \left[ \left| \xi_s \right| \Big| \mathcal{F}_t \right] ds \right)^2 \right] \\
        &\leq 2\|A\|_{C^0([0,T]^3)}^2 \cdot \frac{T^4}{4} \cdot \mathbb{E} \left[ \sup_{0 \leq t \leq T} \mathbb{E} \left[ \sup_{0 \leq s \leq T} \left| \xi_s \right| \Big| \mathcal{F}_t \right]^2 \right] 
         + 2\|B\|_{C^0} \cdot T^2 \cdot \mathbb{E} \left[ \sup_{0 \leq t \leq T} \mathbb{E} \left[ \sup_{0 \leq s \leq T} \left| \xi_s \right| \Big| \mathcal{F}_t \right]^2 \right] \\
        &\leq \left( \frac{a^2 \vee (\eta\phi)}{4\eta^4} \cdot T^4 + \frac{1}{\eta^2} \cdot T^2 \right) \cdot \exp\left(4T\frac{a \vee \sqrt{\eta\phi}}{\eta}\right) \cdot \|\xi\|_{\mathcal{D}}^2.
    \end{align*}}
    Thus, giving the best bound for the operator $L$:
    \begin{equation}\label{bound_of_L}
        \|L\|_{B(\mathcal{D})} \leq \exp\left(2T\frac{a \vee \sqrt{\eta\phi}}{\eta}\right) \cdot \sqrt{\frac{a^2 \vee (\eta\phi)}{4\eta^4} \cdot T^4 + \frac{1}{\eta^2} \cdot T^2}.
    \end{equation}
    Similarly, recall the definition of $\hat{L}$ in formula \eqref{defhatL}. We have
    \begin{equation*}
        \sup_{0 \leq s,t \leq T} \left| \hat{L}_1 \left[E,F\right]_{s,t} \right| \leq \exp\left(2T\frac{a \vee \sqrt{\eta\phi}}{\eta}\right) \cdot \left( \frac{a \vee \sqrt{\eta\phi}}{2\eta^2} \cdot T^2 \cdot \|E\|_{C^0} + \frac{a \vee \sqrt{\eta\phi}}{2\eta^2} \cdot T \cdot \|F\|_{C^0} + \frac{1}{2\eta} \cdot T \cdot \|E\|_{C^0}\right),
    \end{equation*}
    and
    \begin{equation*}
        \sup_{0 \leq t \leq T} \left| \hat{L}_2 \left[E,F\right]_{t} \right| \leq \exp\left(2T\frac{a \vee \sqrt{\eta\phi}}{\eta}\right) \cdot \left( \frac 1{4\eta} \cdot T^2 \cdot \|E\|_{C^0} + \frac 1{2\eta} \cdot T \cdot \|F\|_{C^0}\right).
    \end{equation*}
    Thus, giving the best bound for the operator $L$:
    \begin{equation}\label{bound_of_hat_L}
        \|\hat{L}\|_{B(\mathcal{C})} \leq \exp\left(2T\frac{a \vee \sqrt{\eta\phi}}{\eta}\right) \cdot \left( \frac 1{4\eta} T^2 + \frac 1{2\eta} T \right) \vee \left( \frac{a \vee \sqrt{\eta\phi}}{2\eta^2} (T^2+T) + \frac 1{2\eta} T \right).
    \end{equation}
    In total, we can derive the best bound for $b$, 
    {\small
    $$
    0 < b < \exp\left( - 2T\frac{a \vee \sqrt{\eta\phi}}{\eta}\right) \cdot \left(\frac{a^2 \vee (\eta\phi)}{4\eta^4} \cdot T^4 + \frac{1}{\eta^2} \cdot T^2\right)^{-\frac 12} \wedge \left( \frac 1{4\eta} T^2 + \frac 1{2\eta} T \right)^{-1} \wedge \left( \frac{a \vee \sqrt{\eta\phi}}{2\eta^2} (T^2+T) + \frac 1{2\eta} T \right)^{-1}
    $$}
    which, by taking a convenient version of it, gives \eqref{bound of b}.
\end{proof}

\subsection{Optimal Control of the Informed Broker}\label{Optimal_Control_Broker_Mean_Field_Setting}

The problem of the broker consists in maximizing over $\nu^B$ (adapted to $(\mathcal{F}_t)$) and then over $(\mathcal{F}_t)$ (or equivalently over $\mu_\cdot$) the quantity:
\begin{align*}
&\mathbb{E} \left[ \int_0^T \left\{ Q^B_t \left(b\,\bar{\nu}_t  + \mu \right) + \eta \int_{\mathbb{R}} x^2 m_t(\d x) - \eta^B \left( \nu^B_t \right)^2 - 2\,a^B Q^B_t \left( \nu^B_t - \int_{\mathbb{R}} x\, m_t(\d x) \right) - \phi^B \left( Q^B_t \right)^2 \right\} \d t \right] \\
&= \mathbb{E} \left[ \int_0^T \left\{ Q^B_t \left(b\,\bar{\nu}_t  + \mu_t \right) + \eta \mathbb{E}[\hat{\nu}_t^2 | \mathcal{F}_t] - \eta^B \left( \nu^B_t \right)^2 - 2\,a^B Q^B_t \left( \nu^B_t - \bar{\nu}_t \right) - \phi^B \left( Q^B_t \right)^2 \right\} \d t \right],
\end{align*}
where $Q_0^B$ is the deterministic initial inventory of the broker, and
$$
\d Q^B_t = \left( \nu^B_t - \bar{\nu}_t \right) \d t, \qquad Q^B_0 = Q^B_0.
$$
In the first step, we fix the filtration $(\mathcal{F}_t)$ (and thus the process $(\mu_\cdot)$), and in view of Lemma \ref{lem.hatnu}, the process $(\hat{\nu})$ is optimized with respect to the control $\nu^B$. The value function of the broker is then given by 
$$
u^B(t,q) = \alpha^B_t + \beta^B_t q + \gamma^B_t q^2 / 2,
$$
which satisfies a backward Hamilton-Jacobi equation of the form
$$
\begin{array}{l}
\d u^B_t = H^B_t(q,Du^B_t)\d t + \d M^B_t(q) \qquad \text{in } (0,T) \times \mathbb{R},\\
u^B_T(q) = 0 \qquad \text{in } \mathbb{R},
\end{array}
$$
where
\begin{align*}
H^B_t(q,\xi) &= \min_{\nu^B \in \mathbb{R}} \left\{ - (\nu^B - \bar{\nu}_t) \xi - \left( q \left(b\,\bar{\nu}_t  + \mu_t \right) + \eta \mathbb{E}[\hat{\nu}_t^2 | \mathcal{F}_t] - \eta^B \left( \nu^B \right)^2 - 2\,a^B q \left( \nu^B - \bar{\nu}_t \right) - \phi^B q^2 \right) \right\} \\
&= - q \left(b\,\bar{\nu}_t  + \mu_t \right) - \eta \mathbb{E}[\hat{\nu}_t^2 | \mathcal{F}_t] - 2\,a^B q \bar{\nu}_t + \phi^B q^2 - \frac{| \xi - 2\,a^B q |^2}{4\eta^B} + \bar{\nu}_t \xi,
\end{align*}
which is of the form \eqref{notgenH} with 
$$
A_t = - \eta \mathbb{E}[\hat{\nu}_t^2 | \mathcal{F}_t], \quad B_t = - \left(b\,\bar{\nu}_t  + \mu_t \right) - 2\,a^B \bar{\nu}_t, \quad C_t = 2 \left(\phi^B - \frac{(a^B)^2}{\eta^B}\right), \quad D_t = \bar{\nu}_t, \quad E_t = \frac{a^B}{\eta^B}, \quad F_t = - \frac{1}{2\eta^B}.
$$
This yields the system
$$
\left\{
\begin{array}{ll}
(i) & \d \gamma^B_t = \left( 2 \left( \phi^B - \frac{(a^B)^2}{\eta^B} \right) + 2\frac{a^B}{\eta^B}\gamma^B_t - \frac{1}{2\eta^B} (\gamma^B_t)^2 \right) \d t + \d Z^B_{\gamma,t}, \qquad \gamma^B_T = 0, \\ [1ex]
(ii) & \d \beta^B_t = \left( - \left( b + 2 a^B -\gamma^B_t \right) \bar{\nu}_t  - {\mu}_t + \frac{2a^B - \gamma^B_t}{2\eta^B} \beta^B_t \right) \d t + \d Z^B_{\beta,t}, \qquad \beta^B_T = 0, \\ [1ex]
(iii) & \d \alpha^B_t = \left( - \eta \mathbb{E}[\hat{\nu}_t^2 | \mathcal{F}_t] + \bar{\nu}_t \beta^B_t - \frac{1}{4\eta^B} (\beta^B_t)^2 \right) \d t + \d Z^B_{\alpha,t}, \qquad \alpha^B_T = 0.
\end{array}
\right.
$$
Similarly, the equation for $\gamma^B$ has a unique (deterministic) analytic solution defined on $[0,T]$.
By taking $\Gamma^B_{s,t}$ to be $\exp(\int_s^t (\gamma^B_r - 2a^B) / (2\eta^B) \d r)$, and recalling Lemma \ref{lem.hatnu}, we can write $\beta^B$ as follows:
\begin{equation}\label{defbetaB}
\begin{aligned}
\beta^B_t &= \int_t^T (b + 2a^B - \gamma^B_s) \cdot \Gamma^B_{s,t} \cdot \mathbb{E}[\bar{\nu}_s | \mathcal{F}_t] \d s + \mu_t \int_t^T \Gamma^B_{s,t} \d s \\
          &= \int_t^T (b + 2a^B - \gamma^B_s) \cdot \Gamma^B_{s,t} \cdot \left( (\bar{A}_s + \bar{B}_s) \bar{Q}_0 + \int_0^t \bar{C}_{r,s} \mu_r \d r + \mu_t \int_t^s \bar{C}_{r,s} \d r + D_s \mu_t \right) \d s + \mu_t \int_t^T \Gamma^B_{s,t} \d s \\
          &= \bar{A}_t^B \cdot \bar{Q}_0 + \int_0^t \bar{C}_{s,t}^B \cdot \mu_s \d s + \bar{D}_t^B \cdot \mu_t,
\end{aligned}
\end{equation}
where
\begin{equation}\label{barAB_barCB_barDB_definition}
\left\{
\begin{array}{ll}
    \bar{A}_t^B     & := \int_t^T (b + 2a^B - \gamma^B_s) \cdot \Gamma^B_{s,t} \cdot (\bar{A}_s + \bar{B}_s) \d s, \\
    \bar{C}_{s,t}^B & := \int_t^T (b + 2a^B - \gamma^B_r) \cdot \Gamma^B_{r,t} \cdot \bar{C}_{s,r} \d r, \\
    \bar{D}_t^B     & := \int_t^T (b + 2a^B - \gamma^B_s) \cdot \Gamma^B_{s,t} \cdot \left( \int_t^s \bar{C}_{r,s} \d r + \bar{D}_s \right) \d s + \int_t^T \Gamma^B_{s,t} \d s.
\end{array}
\right.
\end{equation}
Note that $\bar A^B$, $\bar C^B$, and $\bar D^B$ are analytic functions, depending only on $a$, $\eta$, $\phi$, $b$, $a^B$, $\eta^B$, $\phi^B$, $Q_0^B$, $\mathbb E[\mu]$, $\mathbb E[\mu^2]$, $\mathbb E[Q_0]$, $\mathbb E[Q_0^2]$, and $T$. 
Finally, we also have for $\alpha^B$ that
$$
\alpha^B_t = -\E\left[ \int_t^T (- \eta \E[ \hat \nu_s^2\ |\ \mathcal F_s] +\bar\nu_s \beta^B_s -  \frac{1}{4\eta^B} (\beta^B_s)^2)ds\ |\ \mathcal F_t\right]
= \E\left[ \int_t^T ( \eta \hat \nu_s^2 - \bar\nu_s \beta^B_s +  \frac{1}{4\eta^B} (\beta^B_s)^2)ds\ |\ \mathcal F_t\right].
$$
Now, solving the optimization problem in Definition \eqref{majorsol} is equivalent to solving the following problem:\\

We try to find the filtration $(\mathcal{F}_t)_{0\leq t \leq T}$ to guarantee the following four properties:
\begin{enumerate}
    \item $\mathcal{F}_t = \sigma(s \to \hat \nu_s^B ; 0\leq s\leq t)$, $\mathbb{E}\left[ \mu  | \mathcal{F}_0 \right] = \mu_0 = \mathbb{E}[\mu]$.
    \item $
    \left\{
    \begin{array}{rl}
        \alpha_t^B = & \mathbb{E}\left[ \int_t^T \eta \hat \nu_s^2 - \bar\nu_s \beta^B_s + \frac{1}{4\eta^B} (\beta^B_s)^2 d s \bigg| \mathcal{F}_t \right], \\
        \beta_t^B = &  \bar A^B_t + \int_0^t \bar C_{s,t}^B \cdot \mathbb{E}\left[ \mu |\mathcal{F}_s \right] d s + \bar D_t^B \cdot \mathbb{E}\left[ \mu | \mathcal{F}_t\right].
    \end{array}
    \right.
    $
    \item $
    \left\{
    \begin{array}{rl}
        \hat\nu_t^B = & \frac{1}{2\eta^B} \left( (\gamma_t^B - 2 a^B) \hat Q_t^B + \beta_t^B \right), \\
        \hat Q_t^B = &  Q_0^B + \int_0^t (\hat \nu_s^B - \bar \nu_s) d s.
    \end{array}
    \right.
    $
    \item $\underset{\nu^B\in \mathcal{A}^B}{\sup}H^{B,m^{\nu^B}}(\nu^B) = \alpha^B_0 +\beta^B_0 Q_0^B+\gamma^B_0(Q_0^B)^2/2 .$
\end{enumerate}

We can first try to optimize $ u^B(0, Q_0^B) = \alpha^B_0 + \beta^B_0 Q_0^B + \gamma^B_0 (Q_0^B)^2 / 2 $ with respect to $(\mathcal{F}_t)$ and check whether the solution satisfies the four properties. Note that $\beta_0^B$ and $\gamma_0^R$ only depend on $\mu_0$, so we only need to optimize $\alpha_0^B$, that is, we want to compute
\begin{align*}
    \underset{(\mathcal{F}_t)~\text{filtration}}{\sup} & \mathbb{E} \left[ \int_0^T \left\{ \eta \left( \bar{A}_s \bar Q_0 + \bar{B}_s Q_0 + \int_0^s \bar{C}_{r,s} \mathbb{E}[\mu | \mathcal{F}_r] dr + \mathbb{E}[\mu | \mathcal{F}_s] \bar{D}_s \right)^2 \right. \right. \\
    & - \left( \bar{A}_s \bar Q_0 + \bar{B}_s \bar Q_0 + \int_0^s \bar{C}_{r,s} \mathbb{E}[\mu | \mathcal{F}_r] dr + \mathbb{E}[\mu | \mathcal{F}_s] \bar{D}_s \right) \left( \bar{A}_s^B \bar Q_0 + \int_0^s \bar{C}_{r,s}^B \mathbb{E}[\mu | \mathcal{F}_r] dr + \mathbb{E}[\mu | \mathcal{F}_s] \bar{D}_s^B \right) \\
    & \qquad \qquad \qquad \qquad \qquad \qquad \qquad \qquad \qquad \qquad \qquad \left. \left. + \frac{1}{4\eta^B} \left( \bar{A}_s^B \bar Q_0 + \int_0^s \bar{C}_{r,s}^B \mathbb{E}[\mu | \mathcal{F}_r] dr + \mathbb{E}[\mu | \mathcal{F}_s] \bar{D}_s^B \right)^2 \right\} ds \right].
\end{align*}
By taking out a constant depending only on $a$, $\eta$, $\phi$, $b$, $a^B$, $\eta^B$, $\phi^B$, $Q_0^B$, $\mathbb{E}[\mu]$, $\mathbb{E}[\mu^2]$, $\mathbb{E}[Q_0]$, $\mathbb{E}[Q_0^2]$, and $T$, we finally obtain the following problem,
\begin{align*}
    \underset{\substack{(\mathcal{F}_t)~\text{filtration} \\ \text{s.t.} \, \mathbb{E} [\mu | \mathcal{F}_s] = \mu_s}}{\sup} \, & \mathbb{E} \left[ \int_0^T \left\{ \left[ \eta (\bar{D}_s)^2 - \bar{D}_s \bar{D}_s^B + \frac{1}{4\eta^B} (\bar{D}_s^B)^2 \right] \mu_s^2 - \bar{D}_s \mu_s \left( \int_0^s \bar{C}^B_{r,s} \mu_r \, \d r \right) - \bar{D}_s^B \mu_s \left( \int_0^s \bar{C}_{r,s} \mu_r \, \d r \right) \right. \right. \\
    & \qquad \qquad \qquad \qquad + 2\eta \bar{D}_s \mu_s \left( \int_0^s \bar{C}_{r,s} \mu_r \, \d r \right) + \frac{1}{2 \eta^B} \bar{D}_s^B \mu_s \left( \int_0^s \bar{C}^B_{r,s} \mu_r \, \d r \right) \\
    & \qquad \qquad \left. \left. + \eta \left( \int_0^s \bar{C}_{r,s} \mu_r \, \d r \right)^2 - \left( \int_0^s \bar{C}_{r,s} \mu_r \, \d r \right) \left( \int_0^s \bar{C}^B_{r,s} \mu_r \, \d r \right) + \frac{1}{4\eta^B} \left( \int_0^s \bar{C}_{r,s}^B \mu_r \, \d r \right)^2 \right\} \, \d s \right] .
\end{align*}
Note that for any integral of the form 
\[
\mathbb{E}\left[ \int_0^T \left( \int_0^s a_{r,s} \mu_r \, dr \right) \left( \int_0^s b_{r,s} \mu_r \, dr \right) ds \right],
\]
we have it to be equal to
\begin{align*}
      & \mathbb{E} \left[ \int_0^T \left( \int_0^s \int_0^s a_{r,s} b_{u,s} \mu_r \mu_u \, dr \, du \right) ds \right] \\
     =& \mathbb{E}\left[ \int_0^T \left( \int_0^s \int_u^s a_{r,s} b_{u,s} \mu_r \mu_u \, dr \, du +  \int_0^s \int_r^s a_{r,s} b_{u,s} \mu_r \mu_u \, du \, dr \right) ds \right] \\
     =& \mathbb{E} \left[ \int_0^T \mathbb{E}\left[ \int_0^s \mathbb{E}\left[ \int_u^s \left( a_{r,s} b_{u,s} + b_{r,s} a_{u,s} \right) \mu_r \mu_u \, dr \ \bigg| \ \mathcal{F}_u \right] \, du \right] ds \right] \\
     =& \int_0^T \int_0^s \left( \int_u^s a_{r,s} b_{u,s} + b_{r,s} a_{u,s} \, dr \right) \cdot \mathbb{E} [\mu_u^2] \, du \, ds \\ 
     =& \int_0^T \left( \int_u^T \int_u^s a_{r,s} b_{u,s} + b_{r,s} a_{u,s} \, dr \, ds \right) \cdot \mathbb{E} [\mu_u^2] \, du .
\end{align*}
Similarly, we have that 
\[
\mathbb{E}\left[ \int_0^T  a_s \mu_s \left( \int_0^s b_{r,s} \mu_r \, dr \right) ds \right] = \int_0^T \left( \int_r^T a_s b_{r,s} \, ds \right) \cdot \mathbb{E}[\mu_r^2] \, dr.
\]
Hence we can find a function $\mathcal{A} \in C^\omega([0,T])$ such that the integral we want to optimize can be transformed into the following form,
\[
\int_0^T \mathcal{A}'_s \cdot \mathbb{E}[\mu_u^2] \, du = - \mathbb{E}[\mu]^2 \cdot \mathcal{A}_0 - \int_0^T \mathcal{A}_s \, d \mathbb{E}[\mu_s^2],
\]
where $\mathcal{A}'$ is defined as follows:
\begin{equation}\label{A_derivative_definition}
    \begin{aligned}
        \mathcal{A}'_s = &~~ \eta~(\bar{D}_s)^2 - \bar{D}_s \bar{D}_s^B + \frac{1}{4\eta^B}~(\bar{D}_s^B)^2 \\
                      & + 2\eta~\int_s^T \bar{D}_r \bar{C}_{s,r} \, dr - \int_s^T (\bar{D}_r \bar{C}_{s,r}^B + \bar{D}_r^B \bar{C}_{s,r}) \, dr + \frac{1}{2\eta^B}~\int_s^T \bar{D}_r^B \bar{C}_{s,r}^B \, dr\\
                      & + 2\eta~ \int_s^T \int_s^u \bar{C}_{r,u} \bar{C}_{s,u} \, dr \, du - \int_s^T \int_s^u (\bar{C}_{r,u} \bar{C}^B_{s,u} + \bar{C}^B_{r,u} \bar{C}_{s,u}) \, dr \, du + \frac{1}{2\eta^B}~ \int_s^T \int_s^u \bar{C}^B_{r,u} \bar{C}^B_{s,u} \, dr \, du.
    \end{aligned}
\end{equation}
To optimize this integral, we can in fact take,
$$
\mathbb{E}[\mu_t^2] = \left\{ 
\begin{array}{cc}
    \mathbb{E}[\mu]^2 & \text{, if~~} t \leq t_c, \\ 
    \mathbb{E}[\mu^2] & \text{, if~~} t > t_c,
\end{array} 
\right.
$$
where the critical time $t_c$ is determined by the negative minimum of $\mathcal{A}$ as follows:
$$
t_c \in \left\{
\begin{array}{ll}
    \underset{t \in [0,T]}{\operatorname{argmin}} ~ \mathcal{A}_t \wedge 0 & \text{, if~} \exists ~t, \mathcal{A}_t < 0,  \\ 
    (T,\infty) & \text{, if~} \forall ~t, \mathcal{A}_t \geq 0,
\end{array}
\right.
$$
and hence we can take the optimal filtration to be (it is not unique),
$$
\mathcal{F}_t = \left\{ 
\begin{array}{cc}
    \{\emptyset,\Omega\} & \text{, if~~} t \leq t_c, \\ 
    \sigma(\mu) & \text{, if~~} t > t_c.
\end{array} 
\right.
$$
Finally, it is easy to see that such a filtration indeed satisfies properties 1-4 proposed above, which gives us the following theorem.

\begin{theorem}\label{best_control_for_broker}
    Assuming $b>0$ is small enough, there exists $(\hat{A}^B, \hat{B}^B, \hat{C}^B, \hat{D}^B) \in C^{\omega}([0,T] \times [0,T], \mathbb{R}^4)$ such that the optimal control for the broker $\hat{\nu}^B = \hat{\nu}^B(Q_0^B)$ is given by 
    \begin{equation}\label{hatnuB_solution}
        \left\{
        \begin{array}{ll}
            \hat{\nu}^B_t = \hat{A}^B_t \cdot \bar{Q}_0 + \hat{B}^B_t \cdot Q_0^B + \left(\int_0^t \hat{C}^B_{s,t} \, ds + \hat{D}^B_t \right) \cdot \mathbb{E}[\mu] & \text{, when~} t \leq t_c, \\
            \hat{\nu}^B_t = \hat{A}^B_t \cdot \bar{Q}_0 + \hat{B}^B_t \cdot Q_0^B + \int_0^{t_c} \hat{C}^B_{s,t} \, ds \cdot \mathbb{E}[\mu] + \left( \int_{t_c}^t \hat{C}^B_{s,t} \, ds + \hat{D}^B_t \right) \cdot \mu & \text{, when~} t > t_c.
        \end{array}
        \right.
    \end{equation}
    where the critical time $t_c$ is determined by the negative minimum of $\mathcal{A}$ with its derivative defined as in equation \eqref{A_derivative_definition}.
    Note that since $(\bar{C}^B, \bar{D}^B) \in C^{\omega}([0,T] \times [0,T], \mathbb{R}^4)$, we indeed have
    \begin{equation*}
        \left\{
        \begin{array}{lll}
            \sigma(\hat{\nu}^B_s; 0 \leq s \leq t) &= \{\emptyset,\Omega\} & \text{, if~~} t \leq t_c, \\ 
            \sigma(\hat{\nu}^B_s; 0 \leq s \leq t) &= \sigma(\mu) & \text{, if~~} t > t_c.
        \end{array}
        \right.
    \end{equation*}
\end{theorem}

\begin{remark}
    We give in section \ref{nontrivial} an example in which the critical time is nontrivial, i.e., neither equal to $0$, nor equal to $T$.
\end{remark}

\begin{remark}
The expressions for $(\hat A^B, \hat B^B, \hat C^B, \hat D^B)$ in Theorem \ref{best_control_for_broker} are given by the following formulas:
\[
\left\{
\begin{aligned}
\hat A^B_t & = \frac{1}{2\eta^B} \bar A^B_t + \int_0^t \frac{\gamma_t^B-2a^B}{2\eta^B} \Gamma_{s,t}^B \left( \frac{1}{2\eta^B} \bar A^B_s - \bar A_s - \bar B_s \right) ds,\\
\hat B^B_t & = \frac{\gamma_t^B-2a^B}{2\eta^B} \Gamma_{0,t}^B, \\
\hat C^B_{s,t} & = \frac{1}{2 \eta^B} \bar C_{s,t}^B + \frac{\gamma_t^B-2a^B}{2\eta^B} \left( \int_s^t \Gamma_{r,t}^B \left( \frac{1}{2\eta^B} \bar C_{s,r}^B - \bar C_{s,r} \right)dr + \Gamma_{s,t}^B (\frac{1}{2\eta^B} \bar D_s^B - \bar D_s ) \right), \\
\hat D^B_t & = \frac{1}{2 \eta^B} \bar D_{t}^B.
\end{aligned}
\right.
\]
\end{remark}

\section{Stackelberg Problem with Finitely Many Traders}\label{section_comments_finite}
In this section, we investigate the relationship between the problem in the mean field setting and those with finitely many traders, aiming to prove that the optimal control in the mean field setting is approximately optimal in the finite case. Let $N \in \mathbb{N}$ be a large number of traders. We consider Stackelberg equilibria of a differential game in which the $N$ traders interact with the broker. 
Let the probability space $\left(\Omega, \mathcal{F}, \mathbb{P}\right)$ support a family of i.i.d. initial conditions of the traders' inventories $(Q_0^{N,j})_{1 \leq j \leq N}$ following distribution $m_0$ for any $N \geq 1$.

\subsection{Nash Equilibrium of the Traders}
As in Section \ref{section_mfg_equilibrium_for_traders}, in this section, the broker's execution rate \( \nu^B \) is fixed. A Nash equilibrium $\nu^N := (\nu^{N,j})_{1 \leq j \leq N}$ associated with the costs $(H^{N,j}[\nu^B](\cdot))_{1 \leq j \leq N}$ is reached by the traders if
\begin{equation*} 
    H^{N,j}\left[\nu^B\right]\left(\nu^N\right) = \underset{\nu \in \mathcal{A}^{N,j}}{\sup} H^{N,j}\left[\nu^B\right]\left(\nu,(\nu^{N,k})_{k \neq j}\right)
\end{equation*}
for $1 \leq j \leq N$, where
\begin{equation}\label{trader_obj_finite}
    H^{N,j}\left[\nu^B\right]\left(\nu^N\right) = \mathbb{E} \left[ \int_0^T L(Q_t^{N,j}, \nu_t^{N,j}) + F(Q_t^{N,j}, \mathcal{E}_{\nu_t^N}) \, \mathrm{d}t \right].
\end{equation}
with functions $L$, $F$, and empirical measure $\mathcal{E}_{\nu^N}$ defined as follows:
\[
    L(Q, \nu) = - \eta \cdot \nu^2 - 2a \cdot Q \nu - \phi \cdot Q^2,~~
    F(Q, m) = Q \cdot \left(b \int_{\mathbb{R}} \nu \, m(\mathrm{d}\nu) + \mu \right),~~
    \mathcal{E}_{\nu_t^N} = \frac{1}{N} \sum_{j=1}^N \delta_{\nu_t^{N,j}}.
\]
We should also note that we choose the control from the admissible control set $\mathcal{A}^{N,j}$:
\[
    \mathcal{A}^{N,j} = \left\{ \nu \in \mathbf{L}^2([0,T] \times \Omega): \nu_t \text{~is~} \mathcal{F}_t^{N,j}\text{-progressively measurable} \right\},
\]
where $\mathcal{F}_t^{N,j}$ is the filtration generated by the initial condition of the inventory process of the trader $Q_0^{N,j}$ and the execution rate $(\nu^B_t)_{t \in [0,T]}$ of the broker, 
showing that traders should decide their strategies based on the information of their past controls and the initial inventories. We also define $\mathcal{F}_t$ to be $\sigma\left(\nu^B_s : 0 \leq s \leq t\right)$ for convenience.

\subsection{Stackelberg Equilibrium of the Informed Broker}\label{section: broker_finite_problem}
A Stackelberg equilibrium should be attained for the informed broker, that is:
\begin{equation}\label{lem.finite_problem_for_broker}
    \underset{\nu^B \in \mathcal{A}^{N,B}}{\sup} ~\underset{\nu^N \text{~Nash equilibrium}}{\inf} H^{N,B}\left[\mathcal{E}_{\nu^N}\right] \left(\nu^B\right),
\end{equation}
where
\begin{equation}\label{broker_obj_finite}
    H^{N,B}\left[\mathcal{E}_{\nu^N}\right]\left(\nu^B\right) = \mathbb{E} \left[ \int_0^T L^B(Q^B_t[\mathcal{E}_{\nu^N}, \nu^B], \nu^B_t) + F^B(Q^B_t[\mathcal{E}_{\nu^N}, \nu^B], \mathcal{E}_{\nu_t^N}) \, \mathrm{d}t \right].
\end{equation}
with functions $L^B$ and $F^B$ defined as follows:
\[
    L^B(Q, \nu) = - \eta^B \cdot \nu^2 - 2a^B \cdot Q \nu - \phi^B \cdot Q^2,~~
    F^B(Q, m) = Q \cdot \left( (b + 2a^B) \int_{\mathbb{R}} \nu \, m(\mathrm{d} \nu) + \mu\right) + \eta \cdot \int_{\mathbb{R}} \nu^2 \, m(\mathrm{d} \nu).
\]
We also have the inventory process of the broker $Q^B_t[m,\nu^B]$ as:
\[
    Q^B_t[m,\nu^B] = Q_0^B + \int_0^t \left( \nu_s^B - \int_{\mathbb{R}} \nu \, m(\mathrm{d}\nu) \right) \, \mathrm{d}s.
\]
We should also note that we choose the control from the admissible control set $\mathcal{A}^{N,B}$:
\[
    \mathcal{A}^{N,B} = \left\{ \nu \in \mathbf{L}^2([0,T] \times \Omega): \nu_t \text{ is independ of }(Q_0^{N,j})_{1\leq j \leq N} \right\},
\]
and we write \( \mathcal{F}^{N,B} \) for the filtration generated by the control \( \nu^{N,B} \), chosen by the broker in the \( N \)-players game. The independence of the broker's control from the initial conditions of the traders arises from the same reason we have already stated in Section \ref{section_stackelberg_equilibrium_for_broker}.\\

\subsection{Statement of the Main Result}

\begin{theorem}\label{thm.finite_1}
Assume that $m_0 \in \mathcal{P}_4(\mathbb{R})$ and $b$ is small enough. Fix $\epsilon > 0$ and let $\hat{\nu}^B$ be defined as in equation \eqref{hatnuB_solution}, which is a maximizer for the problem in Definition \ref{majorsol}. 
Then there exists $C>0$ and $N_0 \in \mathbb{N}$ such that for all $N \geq N_0$, $\hat{\nu}^B$ is a $C/\sqrt{N}$-optimizer of problem \eqref{lem.finite_problem_for_broker},
$$\underset{\nu^B \in \mathcal{A}^{N,B}}{\sup} ~\underset{\nu^N \text{~Nash equilibrium}}{\inf} H^{N,B}\left[\mathcal{E}_{\nu^N}\right]\left(\nu^B\right),$$
where $C$ is a constant only related to $a$, $\eta$, $\phi$, $b$, $a^B$, $\eta^B$, $\phi^B$, $Q_0^B$, $\mathbb{E}[\mu]$, $\mathbb{E}[\mu^2]$, $\mathbb{E}[Q_0]$, $\mathbb{E}[Q_0^2]$, and $T$.
\end{theorem}
In order to prove the result above, we need the following two lemmas:

\begin{lemma}\label{lem.finite_2}
Under the assumptions of Theorem \ref{thm.finite_1}, there exists $N_0 \in \mathbb{N}$ such that for any admissible control of the broker $\nu^B \in \mathcal{A}^{N,B}$, and any $N \geq N_0$, there exists a unique Nash equilibrium associated with $(H^{N,j}[\nu^B](\cdot))_{1 \leq j \leq N}$.
\end{lemma}
\begin{lemma}\label{lem.finite_3}
Under the assumptions of Theorem \ref{thm.finite_1}, there exists $N_0 \in \mathbb{N}$ such that for any admissible control of the broker $\nu^B$ 
(note that it belongs to both $\mathcal{A}^{N,B}$ and $\mathcal{A}^B$), any $N \geq N_0$, and any Nash equilibrium $\nu^N$ associated with $(H^{N,j}[\nu^B](\cdot))_{1 \leq j \leq N}$, we have,
\begin{equation}\label{difference_between_hatnu_and_empirical_nu}
    \mathbb{E} \left[ \underset{0 \leq t \leq T}{\sup} \left| \frac{1}{N} \sum_{j=1}^{N} \nu_t^{N,j} - \mathbb{E}[\hat{\nu}_t | \mathcal{F}_t] \right|^2 \right] \leq \frac{C}{N},~~
    \mathbb{E} \left[ \underset{0 \leq t \leq T}{\sup} \left| \frac{1}{N} \sum_{j=1}^{N} \left(\nu_t^{N,j}\right)^2 - \mathbb{E}[\hat{\nu}_t^2 | \mathcal{F}_t] \right| \right] \leq \frac{C}{\sqrt{N}},
\end{equation}
where $\hat{\nu}_t$ is defined in equation \eqref{hatnu_solution}, and $C$ is a constant only related to $a$, $\eta$, $\phi$, $b$, $a^B$, $\eta^B$, $\phi^B$, $Q_0^B$, $\mathbb{E}[\mu]$, $\mathbb{E}[\mu^2]$, $\mathbb{E}[Q_0]$, $\mathbb{E}[Q_0^2]$, and $T$.
\end{lemma}

\begin{proof}[Proof of Theorem \ref{thm.finite_1}]

    Throughout the proof, we assume $N \geq N_0$, where $N_0$ is as defined in Lemma \ref{lem.finite_3}. \\

    (i) First, we estimate the inventory process of the broker. By the first inequality in \eqref{difference_between_hatnu_and_empirical_nu}, for any Nash equilibrium $\nu^N$ associated with $(H^{N,j}[\nu^B](\cdot))_{1 \leq j \leq N}$, we can derive:
    \begin{equation}\label{difference_between_QNB_and_QB}
        \mathbb{E} \left[ \underset{0 \leq t \leq T}{\sup} \left| Q_t^{N,B} - Q^B_t \right|^2 \right] \leq \frac{C}{N},
    \end{equation}
    where
    $$
    Q_t^{N,B} := Q^B_t[\mathcal{E}_{\nu^N}, \nu^B],~~ Q_t^{B} := Q^B_t[m^{\nu^B}, \nu^B].
    $$
    Moreover, we also know that 
    \begin{align*}
        \|Q^B\|^2_{\mathbf{L}^2 ([0,T] \times \Omega)} &\leq 2T \cdot (Q_0^B)^2 + 2 \cdot \int_0^T \mathbb{E} \left[ \left| \int_0^t (\nu_s^B - \bar{\nu}_s) \, ds \right|^2 \right] dt \\
        &\leq C + C \cdot \|\nu^B - \bar{\nu}\|^2_{\mathbf{L}^2 ([0,T] \times \Omega)} \leq C \cdot \left( 1 + \|\nu^B\|^2_{\mathbf{L}^2 ([0,T] \times \Omega)} \right),
    \end{align*}
    where the last inequality comes from \eqref{first_second_conditional_moment_of_hatnu}, which shows that $\|\bar{\nu}\|_{\mathbf{L}^2 ([0,T] \times \Omega)}$ is always bounded by a constant only related to $a$, $\eta$, $\phi$, $b$, $a^B$, $\eta^B$, $\phi^B$, $Q_0^B$, $\mathbb{E}[\mu]$, $\mathbb{E}[\mu^2]$, $\mathbb{E}[Q_0]$, $\mathbb{E}[Q_0^2]$, and $T$. Hence, we have:
    \begin{equation}\label{QB_estimate}
        \|Q^B\|_{\mathbf{L}^2 ([0,T] \times \Omega)} \leq C \left( 1 + \|\nu^B\|_{\mathbf{L}^2 ([0,T] \times \Omega)} \right).
    \end{equation}

    (ii) Next, we prove that 
    \begin{equation}\label{HNB_estimate}
        \left| H^{N,B}\left[\mathcal{E}_{\nu^N}\right]\left(\nu^B\right) - H^{N,B}\left[m^{\nu^B}\right]\left(\nu^B\right) \right| \leq \frac{C}{\sqrt{N}} \left( 1 + \|\nu^B\|_{\mathbf{L}^2 ([0,T] \times \Omega)} \right),
    \end{equation}
    where $\nu^B \in \mathcal{A}^{N,B}$ is an admissible control for the major player, 
    $(\mathcal{E}_{\nu_t^N})_{t \in [0,T]}$ is the empirical distribution of $\nu^N$ at time $t$, which attains a Nash equilibrium in the finite game, 
    and $(m^{\nu^B}_t)_{t \in [0,T]}$ is the distribution of $\hat{\nu}$ conditional to $\mathcal{F}_t := \{ s \to \nu^B_s ; 0 \leq s \leq t \}$ at time $t$.    
    Here, $C$ is a constant only depending on $a$, $\eta$, $\phi$, $b$, $a^B$, $\eta^B$, $\phi^B$, $Q_0^B$, $\mathbb{E}[\mu]$, $\mathbb{E}[\mu^2]$, $\mathbb{E}[Q_0]$, $\mathbb{E}[Q_0^2]$, and $T$. \\

    By formula \eqref{broker_obj_infinite}, $H^{N,B}$ is composed of two parts, $L^B$ and $F^B$. Hence, we will estimate $H^{N,B}$ according to these two parts. 
    Combining inequalities \eqref{difference_between_QNB_and_QB}-\eqref{QB_estimate}, for $L^B$, we have
    \begin{align*}
        & \left| \mathbb{E} \left[ \int_0^T L^B\left( Q_t^{N,B}, \nu^B_t \right) - L^B\left( Q_t^B, \nu^B_t \right) \, dt \right] \right| 
    \leq C \cdot \mathbb{E} \left[ \int_0^T \left| ( Q_t^{N,B} )^2 - ( Q_t^B )^2 \right| + \left| \left( Q_t^{N,B} - Q_t^B \right) \nu_t^B \right| \, dt \right] \\
    &\qquad \qquad \qquad \qquad \leq C \cdot \mathbb{E} \left[ \int_0^T \left| Q_t^{N,B}  - Q_t^B \right| \cdot \left( | Q_t^{N,B} | + | Q_t^B | + | \nu_t^B | \right) \, dt \right] \\
    &\qquad \qquad \qquad \qquad \leq C \cdot \|Q^{N,B} - Q^B\|_{\mathbf{L}^2 ([0,T] \times \Omega)} \left( \|Q^{N,B} \|_{\mathbf{L}^2 ([0,T] \times \Omega)} + \|Q^B\|_{\mathbf{L}^2 ([0,T] \times \Omega)} + \|\nu^B\|_{\mathbf{L}^2 ([0,T] \times \Omega)} \right) \\
    &\qquad \qquad \qquad \qquad \leq \frac{C}{\sqrt{N}} \left( 1 + \|\nu^B\|_{\mathbf{L}^2 ([0,T] \times \Omega)} \right),
    \end{align*} 
    while for $F^B$, combining inequalities \eqref{first_second_conditional_moment_of_hatnu}, and \eqref{difference_between_hatnu_and_empirical_nu}-\eqref{QB_estimate}, we have
    \begin{align*}
    &\left| \mathbb{E} \left[ \int_0^T F^B\left( Q_t^{N,B}, \mathcal{E}_{\nu_t^N} \right) - F^B\left( Q_t^B, m^{\nu^B}_t \right) \, dt \right] \right| \\
    &\leq C \cdot \mathbb{E} \left[ \int_0^T \left| \frac{1}{N} \sum_{j=1}^{N} \left( \nu_t^{N,j} \right)^2 - \mathbb{E}[\hat{\nu}_t^2 | \mathcal{F}_t] \right| \, dt \right] \\
    &\qquad\qquad\qquad+ C \cdot \mathbb{E} \left[ \int_0^T \left| Q_t^{N,B} \left((b + 2a^B) \cdot \frac{1}{N} \sum_{j=1}^{N} \nu_t^{N,j} + \mu \right) - Q_t^B \left((b + 2a^B) \cdot \bar{\nu}_t + \mu \right) \right| \, dt \right] \\
    &\leq \frac{C}{\sqrt{N}} + C \cdot \mathbb{E} \left[ \int_0^T \left| Q_t^{N,B} - Q_t^B \right| \cdot \left( |\bar{\nu}_t| + |\mu| \right) + \left| Q_t^{N,B} \right| \cdot \left| \frac{1}{N} \sum_{j=1}^{N} \nu_t^{N,j} - \bar{\nu}_t \right| \, dt \right] \\
    &\leq \frac{C}{\sqrt{N}} + C \cdot \|Q^{N,B} - Q^B\|_{\mathbf{L}^2 ([0,T] \times \Omega)} \left( \|\bar{\nu}\|_{\mathbf{L}^2 ([0,T] \times \Omega)} + \|\mu\|_{\mathbf{L}^2 (\Omega)} \right) + C \cdot \big\| \frac{1}{N} \sum_{j=1}^{N} \nu_t^{N,j} - \bar{\nu}_t \big\|_{\mathbf{L}^2 ([0,T] \times \Omega)} \|Q^{N,B}\|_{\mathbf{L}^2 ([0,T] \times \Omega)} \\
    &\leq \frac{C}{\sqrt{N}} \left( 1 + \|\bar{\nu}\|_{\mathbf{L}^2 ([0,T] \times \Omega)} + \|\mu\|_{\mathbf{L}^2 (\Omega)} + \|Q^{N,B}\|_{\mathbf{L}^2 ([0,T] \times \Omega)} \right) \\
    &\leq \frac{C}{\sqrt{N}} \left( 1 + \|\nu^B\|_{\mathbf{L}^2 ([0,T] \times \Omega)} \right).
    \end{align*} 
    Hence, in the end, we can derive equation \eqref{HNB_estimate}. \\

    (iii) Now, we try to prove that if $\mathcal{E}_{\nu^N}$ attains the Nash equilibrium given $\nu^B$, then we have the fact that the optimizer of problem \eqref{lem.finite_problem_for_broker} should ensure
    \begin{equation}\label{a-priori_estimate_of_nuB}
        \|\nu^B\|_{\mathbf{L}^2([0,T] \times \Omega)} \leq C,
    \end{equation}
    where $C$ is a constant related only to $a$, $\eta$, $\phi$, $b$, $a^B$, $\eta^B$, $\phi^B$, $Q_0^B$, $\mathbb{E}[\mu]$, $\mathbb{E}[\mu^2]$, $\mathbb{E}[Q_0]$, $\mathbb{E}[Q_0^2]$, and $T$.\\

    Note that, by applying the Cauchy-Schwarz inequality, we have
    \begin{align*}
        H^{N,B}&\left[\mathcal{E}_{\nu^N}\right]\left(\nu^B\right)  = \mathbb{E} \left[ - a^B \left( (Q_T^{N,B})^2 - (Q_0^{N,B})^2 \right) + \int_0^T - \eta^B (\nu^{B}_t)^2 - \phi^B (Q^{N,B}_t)^2 + F^B(Q_t^{N,B}, \mathcal{E}_{\nu_t^N}) \, \mathrm{d}t \right] \\
        & \leq a^B \cdot (Q_0^{B})^2 + \mathbb{E} \left[ \int_0^T - \eta^B (\nu^{B}_t)^2 - \frac{\phi^B}{2} (Q^{N,B}_t)^2 + \frac{1}{2\phi^B} \left( (b + 2a^B) \frac{1}{N} \sum_{j=1}^{N} \nu_t^{N,j} +\mu \right)^2 + \eta \frac{1}{N} \sum_{j=1}^{N} (\nu_t^{N,j})^2 \, \mathrm{d}t \right]\\
        & \leq - \eta^B \cdot \|\nu^B\|^2_{\mathbf{L}^2([0,T] \times \Omega)} + C \cdot \left( 1 + \mathbb{E} \left[ \int_0^T \left( \frac{1}{N} \sum_{j=1}^{N} \nu_t^{N,j}\right)^2 + \frac{1}{N} \sum_{j=1}^{N} (\nu_t^{N,j})^2 \, \mathrm{d}t \right] \right) \\
        & \leq - \eta^B \cdot \|\nu^B\|^2_{\mathbf{L}^2([0,T] \times \Omega)} + C \cdot \left( 1 + \|\bar \nu\|^2_{\mathbf{L}^2([0,T] \times \Omega)} + \|\bar M\|_{\mathbf{L}^1([0,T] \times \Omega)} \right) \\
        & \leq - \eta^B \cdot \|\nu^B\|^2_{\mathbf{L}^2([0,T] \times \Omega)} + C,
    \end{align*}
    where the last inequality is derived from equation \eqref{first_second_conditional_moment_of_hatnu}, which shows that $\|\bar \nu\|_{\mathbf{L}^2 ([0,T] \times \Omega)}$ and $\|\bar M\|_{\mathbf{L}^1 ([0,T] \times \Omega)}$ are always bounded by a constant related only to $a$, $\eta$, $\phi$, $b$, $a^B$, $\eta^B$, $\phi^B$, $Q_0^B$, $\mathbb{E}[\mu]$, $\mathbb{E}[\mu^2]$, $\mathbb{E}[Q_0]$, $\mathbb{E}[Q_0^2]$, and $T$. Hence, we know that the optimizer of problem \eqref{lem.finite_problem_for_broker} should satisfy inequality \eqref{a-priori_estimate_of_nuB}.\\

    (iv) Now we complete the proof of Theorem \ref{thm.finite_1}. Let $\nu^B$ be any admissible control ensuring $\|\nu^B\|_{\mathbf{L}^2([0,T] \times \Omega)} \leq C$, which satisfies the estimate \eqref{a-priori_estimate_of_nuB} for the optimizer of problem \eqref{lem.finite_problem_for_broker} by (ii), and let $\hat \nu^B$ be the one defined in equation \eqref{hatnuB_solution}, which is also an admissible control for the broker and a maximizer for problem $H^{B}[m^{\nu^B}] (\nu^B)$ by Theorem \ref{best_control_for_broker}. Then, by Lemma \ref{lem.finite_2}, there exist $\nu^N$ and $\hat \nu^N$ unique Nash equilibra for $(H^{N,j}[\nu^B](\cdot))_{1 \leq j \leq N}$ and $(H^{N,j}[\hat \nu^B](\cdot))_{1 \leq j \leq N}$ respectively. Hence
    \begin{align*}
        H^{N,B}\left[\mathcal{E}_{\nu^N}\right]\left(\nu^B\right) &\leq H^{N,B}\left[m^{\nu^B}\right]\left(\nu^B\right) + \frac{C}{\sqrt{N}} = H^{B}\left[m^{\nu^B}\right]\left(\nu^B\right) + \frac{C}{\sqrt{N}} \\
        & \leq H^{B}\left[m^{\hat \nu^B}\right]\left(\hat \nu^B\right) + \frac{C}{\sqrt{N}} = H^{N,B}\left[m^{\hat \nu^B}\right]\left(\hat \nu^B\right) + \frac{C}{\sqrt{N}} \leq H^{N,B}\left[\mathcal{E}_{\hat \nu^N}\right]\left(\hat \nu^B\right) + \frac{C}{\sqrt{N}},
    \end{align*}
    where the first and third inequalities come from equation \eqref{HNB_estimate}, the two equalities come from equations \eqref{broker_obj_infinite} and \eqref{broker_obj_finite}, and the second inequality comes from Theorem \ref{best_control_for_broker}.
\end{proof}

\subsection{Proof of Lemma \ref{lem.finite_2} and \ref{lem.finite_3}}

\begin{proof}[Proof of Lemma \ref{lem.finite_2}]

First, note that, given an admissible control of the broker $\nu^B$, for any admissible control of traders $\nu^N$, we have
\begin{equation*}
    H^{N,j} [\nu^B] (\nu^N) = \mathbb{E} \left[ \int_0^T L(Q_t^{N,j},\nu^{N,j}_t) + Q_t^{N,j} \mathbb{E} [\mu | \mathcal{F}^{N,j}_t ] + b \cdot Q_t^{N,j} \cdot \left( \frac{1}{N} \nu^{N,j}_t + \frac{1}{N} \sum_{k \neq j} \mathbb{E} [\nu^{N,k}_t | \mathcal{F}_t^{N,j} ] \right) \, \mathrm{d}t \right].
\end{equation*}
Since $\mu$ is independent of $Q_0^{N,j}$, $\mathcal{F}_t \subset \sigma(\mu)$, and $\nu^{N,k}_t$ is $\mathcal{F}_t \vee \sigma(Q_0^{N,k})$-progressively measurable, we have 
\[
\mathbb{E} [\mu | \mathcal{F}^{N,j}_t ] = \mathbb{E} [\mu | \mathcal{F}_t ] = \mu_t, ~~
\bar \nu^{N,\hat{j}}_t := \frac{1}{N} \sum_{k \neq j} \mathbb{E} [\nu^{N,k}_t | \mathcal{F}_t^{N,j} ] = \frac{1}{N} \sum_{k \neq j} \mathbb{E} [\nu^{N,k}_t | \mathcal{F}_t ].
\]
Now, by setting $a_N := a - \frac{b}{2N}$ and $N_0 \geq \frac{b}{a}$, we can transform the optimization problem
\[
\left\{
    \begin{aligned}
        & Q_t^{N,j} = Q_0^{N,j} + \int_0^t \nu_s \, \mathrm{d}s, \\
        & \underset{\nu \in \mathcal{A}^{N,j}}{\sup} H^{N,j}[\nu^B] (\nu, (\nu^{N,k})_{k \neq j}),
    \end{aligned}
\right.
\]
into solving an HJ equation
\begin{equation*}
    \left\{
    \begin{aligned}
        \mathrm{d}u_t(q) &= H_t(q, D u_t (q)) \, \mathrm{d}t + \mathrm{d} M_t(q), \\
        u_T(q) &= 0,
    \end{aligned}
    \right.
\end{equation*}
where $H_t(q, \xi)$ is defined as follows:
\begin{align*}
    H_t(q, \xi) &= \min_{\nu \in \mathbb{R}} \left\{  -\nu\xi + \eta \nu^2 + 2 a_N q \nu + \phi q^2 - q (b \bar \nu^{N,\hat{j}}_t + \mu_t)  \right\} \\
    &= - q (b \bar \nu^{N,\hat{j}}_t + \mu_t) + \phi q^2 - \frac{1}{4\eta} \left| \xi - 2 a_N q \right|^2.
\end{align*}
Hence, via a similar process in the proof of Lemma \ref{lem.hatnu}, we find the value function to be
\[
u_t(q) = \alpha^{N,j}_t + \beta^{N,j}_t q + \gamma^{N,j}_t \frac{q^2}{2},
\]
where $\beta^{N,j}$ and $\gamma^{N,j}$ can be determined by
\begin{equation}
    \left\{
        \begin{aligned}
            &\mathrm{d} \gamma_t^N = \left( \phi - \frac{a_N^2}{\eta} + 2 \frac{a_N}{\eta} \gamma_t^N - \frac{1}{2\eta} (\gamma_t^N)^2  \right) \, \mathrm{d}t, ~~ \gamma_T^N = 0, \\
            & \beta_t^{N,j} = b \int_t^T \Gamma^N_{s,t} \mathbb{E}[\bar \nu^{N,\hat{j}}_s | \mathcal{F}_t] \, \mathrm{d}s + \mu_t \int_t^T \Gamma^N_{s,t} \, \mathrm{d}s, \\
            & \Gamma^N_{s,t} := \exp\left(\int_s^t \frac{\gamma_r^N - 2 a_N}{2\eta} \, \mathrm{d}r\right).
        \end{aligned}
    \right.
\end{equation}

We also have the estimate for $\gamma^N$ and hence for $\Gamma^N$
\begin{equation}\label{estimate_for_gammaN}
    \sup_{0 \leq t \leq T} \left| \gamma_t^N - \gamma_t \right| \leq \frac{C}{N},~~ \sup_{0 \leq s,t \leq T} \left| \Gamma_{s,t}^N - \Gamma_{s,t} \right| \leq \frac{C}{N}.
\end{equation}
By the maximum principle applied to equation \eqref{HJexp} (\citeauthor{bismut1976linear} \cite{bismut1976linear}), the optimal control is uniquely determined by
\begin{equation}
    \left\{
        \begin{aligned}
            & Q_t^{N,j} = \nu_t^{N,j} \, \mathrm{d}t,\\
            & \nu_t^{N,j} = - D_\xi H_t (Q_t^{N,j}, D u_t(Q_t^{N,j})) = \frac{1}{2\eta} \left( (\gamma_t^N - 2 a_N) Q_t^{N,j} + \beta_{t}^{N,j} \right).
        \end{aligned}
    \right.
\end{equation}
Hence, we would have
\[
    Q_t^{N,j} = \Gamma^N_{0,t} Q_0^{N,j} + \frac{1}{2\eta} \int_0^t \Gamma^N_{s,t} \beta_t^{N,j} \, \mathrm{d}s,
\]
and
\begin{equation}\label{noname1}
\begin{aligned}
    \nu_t^{N,j} = & \frac{\gamma_t^N - 2 a_N}{2\eta} \Gamma^N_{0,t} \cdot Q_0^{N,j} + \left[  \mu_t \int_t^T \frac{1}{2\eta} \Gamma^N_{s,t} \, \mathrm{d}s + \int_0^t \mu_s \left( \frac{\gamma_t^N - 2 a_N}{4\eta^2} \Gamma^N_{s,t} \int_s^T \Gamma^N_{r,s} \, \mathrm{d}r \right) \, \mathrm{d}s \right] \\
    & + b \left[  \int_0^t \int_s^T  \frac{\gamma_t^N - 2 a_N}{4\eta^2} \Gamma^N_{s,t} \Gamma^N_{r,s} \mathbb{E}[\bar \nu_r^{N,\hat{j}} | \mathcal{F}_s] \, \mathrm{d}r \, \mathrm{d}s + \int_t^T \frac{1}{2\eta} \Gamma^N_{s,t} \mathbb{E}[\bar \nu_s^{N,\hat{j}} | \mathcal{F}_t] \, \mathrm{d}s \right] \\
    = & z_t^N Q_0^{N,j} + \tilde L_N \mu_t + b \cdot \frac{N}{N-1} \cdot L_N \bar \nu_t^{N,\hat{j}},
\end{aligned}    
\end{equation}
with $z^N$, $L_N$, and $\tilde L_N$ defined as follows:
\begin{equation*}
    \left\{
        \begin{aligned}
            & z_t^N = \frac{\gamma_t^N - 2 a_N}{2\eta} \Gamma^N_{0,t}, \\
            & L_N \xi_t = \frac{N-1}{N} \left[ \int_0^t \int_s^T  \frac{\gamma_t^N - 2 a_N}{4\eta^2} \Gamma^N_{s,t} \Gamma^N_{r,s} \mathbb{E}[\xi_r | \mathcal{F}_s] \, \mathrm{d}r \, \mathrm{d}s + \int_t^T \frac{1}{2\eta} \Gamma^N_{s,t} \mathbb{E}[\xi_s | \mathcal{F}_t] \, \mathrm{d}s \right], \\
            & \tilde L_N \xi_t = \xi_t \int_t^T \frac{1}{2\eta} \Gamma^N_{s,t} \, \mathrm{d}s + \int_0^t \xi_s \left( \frac{\gamma_t^N - 2 a_N}{4\eta^2} \Gamma^N_{s,t} \int_s^T \Gamma^N_{r,s} \, \mathrm{d}r \right) \, \mathrm{d}s.
        \end{aligned}
    \right.
\end{equation*}
Hence, if $\nu^N$ is indeed a Nash equilibrium given $\nu^B$, by summing the above equation over $j$ and dividing by $N$, we have
\begin{equation}\label{hatnuN_solution}
    \hat \nu^N_t = z_t^N \hat Q_0^N + \tilde L_N \mu_t + b \cdot L_N \bar \nu_t^N,
\end{equation}
where $\hat \nu^N$, $\bar \nu^N$, and $\hat Q_0^N$ are defined by
\[
\hat \nu^N_t = \frac{1}{N} \sum_{j=1}^N \nu_t^{N,j},~~ \bar \nu^N_t = \mathbb{E}[\hat \nu_t^N | \mathcal{F}_t],~~ \hat Q^N_0 = \frac{1}{N} \sum_{j=1}^N Q_0^{N,j}.
\]
By the definition and equation \eqref{estimate_for_gammaN}, we also have the estimate for $z^N$, $L^N$, and $\tilde L^N$, that is
\begin{equation}\label{estimate_for_LN}
    \sup_{0 \leq t \leq T} \left| z^N_t - z_t \right| \leq \frac{C}{N},~~\left\| \left( L_N - L \right) \xi \right\|_{\mathcal{D}} \vee \left\| \left( \tilde L_N - \tilde L \right) \xi \right\|_{\mathcal{D}} \leq \frac CN \left\| \xi \right\|_{\mathcal{D}}.
\end{equation}
Now, by applying conditional expectation with respect to $\mathcal{F}_t$ on both sides of equation \eqref{hatnuN_solution}, we have
\begin{equation*}
    \bar \nu^N_t = z_t^N \bar Q_0 + \tilde L_N \mu_t + b \cdot L_N \bar \nu_t^N,
\end{equation*}
which, through a similar process as in the proof of Lemma \ref{lem.hatnu}, when $N_0 \geq C/(b^{-1} - \|L\|_{B(\mathcal{D)}} \vee \|\hat L\|_{B(\mathcal{C})})$, guarantees a unique solution as follows:
\begin{equation}\label{barnuN_definition}
    \bar \nu^N_t = \bar A_t^N \bar Q_0 + \bar B_t^N \bar Q_0 + \int_0^t \bar C_{s,t}^N \mu_s \, \mathrm{d}s + \bar D_t^N \mu_t,
\end{equation}
where $(\bar A^N, \bar B^N, \bar C^N, \bar D^N) \in C^{\omega}([0,T] \times [0,T], \mathbb{R}^4)$, and we have the estimate
\begin{equation}\label{estimate_for_ABCDN}
    \sup_{0 \leq s,t \leq T} \max \left\{ \left| \bar A^N_t - \bar A_t \right|, \left| \bar B^N_t - \bar B_t \right|, \left| \bar C^N_{s,t} - \bar C_{s,t} \right|, \left| \bar D^N_t - \bar D_t \right| \right\} \leq \frac{C}{N}.
\end{equation}
Hence, when $N_0 \geq \max \{ b/a, C/(b^{-1} - \|L\|_{B(\mathcal{D})} \vee \|\hat L\|_{B(\mathcal{C})}) \}$, if $\nu^N$ is indeed a Nash equilibrium given $\nu^B$, $\bar \nu^N_t$ must be given by \eqref{barnuN_definition}. Now, since we also have
\begin{equation}\label{noname2}
    \bar \nu_t^{N, \hat{j}} = \frac{1}{N} \sum_{k \neq j} \mathbb{E}[\nu^{N,k}_t | \mathcal{F}_t] = \bar \nu^N_t - \frac{1}{N} \mathbb{E}[\nu^{N,j}_t | \mathcal{F}_t] =: \bar \nu^N_t - \frac{1}{N} \bar \nu_t^{N,j},
\end{equation}
by combining equations \eqref{noname1} and \eqref{noname2}, we can then derive
\begin{equation}
    \bar \nu^{N,j}_t = z_t^N \bar Q_0 + \tilde L_N \mu_t + b \cdot \frac{N}{N-1} \cdot L_N \bar \nu_t^N - \frac{b}{N-1} \cdot L_N \bar \nu_t^{N,j},
\end{equation}
which shows that $\bar \nu^{N,j}$ can also be uniquely determined by 
\begin{equation*}
    \bar \nu^{N,j} = (I + \frac{b}{N-1} L_N )^{-1} \left( Z^N \bar Q_0 + \tilde L_N \mu + b \cdot \frac{N}{N-1} \cdot L_N \bar \nu^N \right).
\end{equation*}
Finally, we will have $\nu_t^{N,j}$ uniquely determined by 
\begin{equation*}
    \nu_t^{N,j} = z_t^N Q_0^{N,j} + \tilde L_N \mu_t + b \cdot \frac{N}{N-1} \cdot L_N \bar \nu_t^N - b \cdot \frac{1}{N-1} \cdot L_N \bar \nu_t^{N,j},
\end{equation*}
which shows the uniqueness and existence of the Nash equilibrium of traders given $\nu^B$. In fact, we can write down the formula of $\nu^N$ by $(\bar A^N, \bar B^N, \bar C^N, \bar D^N)$:
\begin{equation}\label{nuNj_definition}
    \nu^{N,j}_t = \bar A_t^N \bar Q_0 + \bar B_t^N Q_0^{N,j} + \int_0^t \bar C_{s,t}^N \mu_s \, \mathrm{d}s + \bar D_t^N \mu_t,
\end{equation}
which is indeed a Nash equilibrium associated with $(H^{N,j}[\nu^B](\cdot))_{1 \leq j \leq N}$ and is also a unique one.

\end{proof}

\begin{proof}[Proof of Lemma \ref{lem.finite_3}]

According to the Cauchy-Schwarz inequality, the independence of $\left(Q_0^{N,j}\right)_{1 \leq j \leq N}$, equations \eqref{first_second_conditional_moment_of_hatnu}, \eqref{estimate_for_ABCDN}, and \eqref{nuNj_definition}, it can be easily proved.

\end{proof}

\section{Numerical example}\label{nontrivial}

Throughout this section, we assume that $\eta\phi = a^2$ and $\eta^B\phi^B = (a^B)^2$, so that both $\gamma$ and $\gamma^B$ vanish. The parameters used here are summarized in Table~\ref{tab:params_example1}. Note that an important consequence of our result is that the optimal strategy does not depend on the distribution of $\mu$, but only on its expectation and realized value. 

\begin{table}[h]
\centering
\begin{tabular}{|c|c||c|c|}
\hline
Parameter & Value & Parameter & Value \\
\hline
$a$       & $1 \times 10^{-2}\ \$ $   & $a^B$       & $1 \times 10^{-2} \ \$ $ \\
$\eta$    & $1 \times 10^{-2}\ \$ \cdot \text{day}$   & $\eta^B$    & $5 \times 10^{-3}\ \$ \cdot \text{day}$ \\
$T$       & $2\ \text{days}$                  &       $\mathbb E [\mu]$ &  $0\ \$ \cdot \text{day}^{-1}$                       \\
 & & $\mu$ (realized) & $5\ \$ \cdot \text{day}^{-1}$\\
\hline
\end{tabular}
\caption{Parameter values.}
\label{tab:params_example1}
\end{table}

Our first goal is to show that, under this setting and for sufficiently small $b$, the critical time defined in Theorem~\ref{best_control_for_broker} is non-trivial, that is, it differs from both $0$ and $T$.\\

We begin by computing the expressions for $\beta^1$, $\beta^2$, and $\beta^3$. According to formula~\eqref{beta_parameter_definition}, we find:
\begin{equation}\label{beta_parameter_example}
    \beta_s^1,~\beta_s^2 = O(b), \qquad \beta_s^3 = \int_s^T \Gamma_{r,s} \, \mathrm{d}r + O(b) = e^{T-s} -1 + O(b).
\end{equation}
Hence, $\bar C_{s,t}$ and $\bar D_t$ in~\eqref{barA_barB_barC_barD_definition} become:
\begin{equation}\label{barC_barD_example}
    \bar C_{s,t} = 50(e^{s-t} - e^{T-t}) + O(b), \qquad \bar D_t = 50(e^{T-t} - 1) + O(b).
\end{equation}
Similarly, $\bar C^B_{s,t}$ and $\bar D^B_t$ in~\eqref{barAB_barCB_barDB_definition} satisfy:
\begin{equation}\label{barCB_barDB_example}
    \left\{
    \begin{aligned}
        \bar C^B_{s,t} & = -e^{2T-2t} + e^{T + s - 2t} + e^{T - t} - e^{s - t} + O(b), \\ 
        \bar D^B_t     & = \left(-T + t + \frac{5}{2}\right) e^{2T-2t} - 3 e^{T - t} + \frac{1}{2} + O(b).
    \end{aligned}
    \right.
\end{equation}

We then compute $\mathcal{A}^\prime$ as defined in~\eqref{A_derivative_definition}:
\begin{align*}
    \mathcal{A}^\prime_s = & 
    50 \left(T - s - \frac{5}{2} \right) \left( T - s - 2 \right) e^{4(T-s)}
    + 275 \left( T - s - \frac{13}{6} \right) e^{3(T-s)} \\
    & - 50 \left( T - s - \frac{107}{12} \right) e^{2(T-s)}
    - \frac{225}{2} e^{T-s}
    - 50 \left( T - s - \frac{17}{12} \right) 
    + 50 \left( T - s - \frac{13}{6} \right) e^{-(T-s)}
    + 50 e^{-2(T-s)} + O(b).
\end{align*}

\begin{figure}[!htbp]
    \centering
    \includegraphics[width=0.55\textwidth]{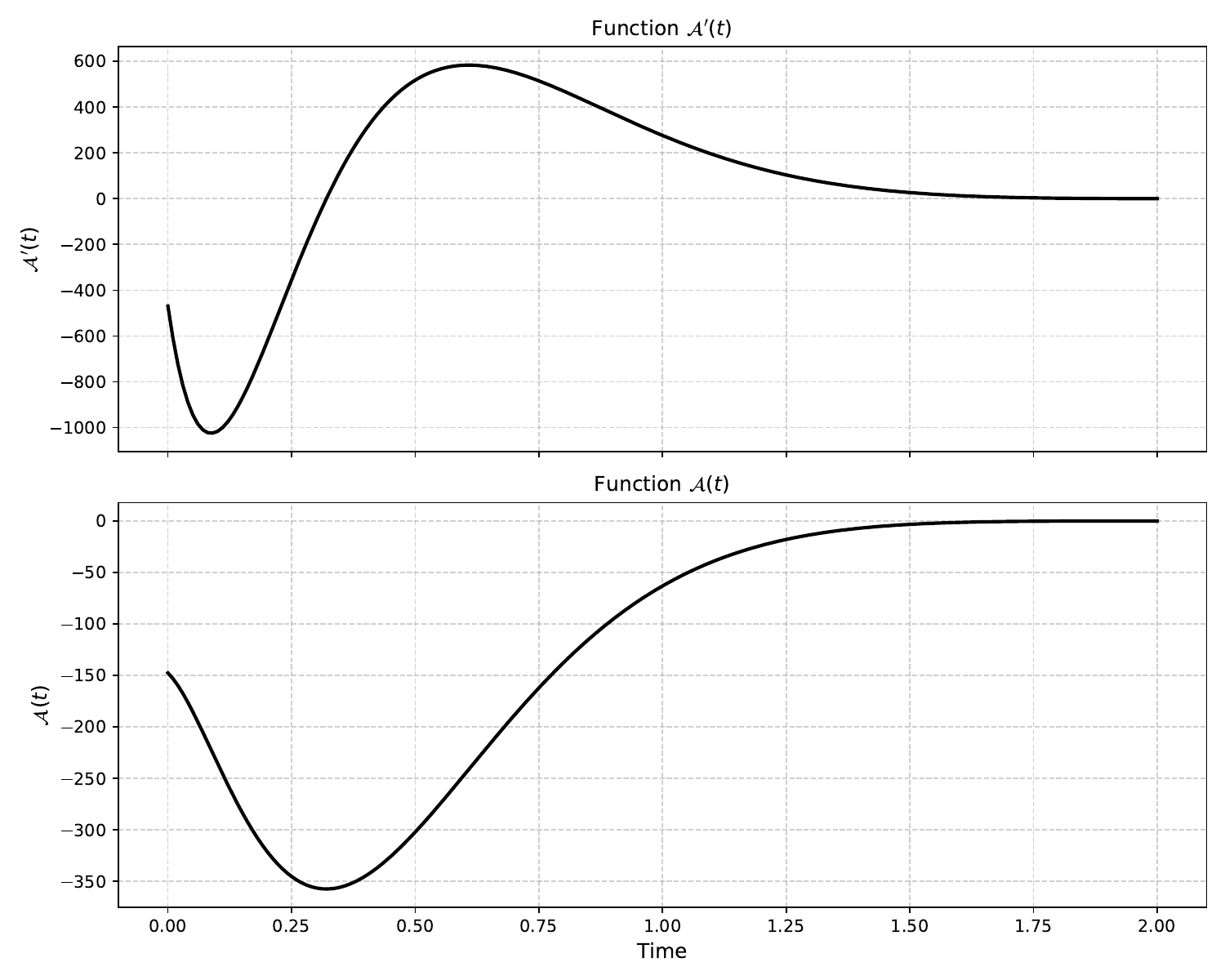}
    \caption{The main terms of $\mathcal{A}^\prime$ and $\mathcal A$.}
    \label{fig:A_functions}
\end{figure}

Figure~\ref{fig:A_functions} displays the functions $\mathcal{A}^\prime$ and its antiderivative $\mathcal{A}$, neglecting the $O(b)$ term. We observe that, as long as $b$ is sufficiently small, the function $\mathcal{A}$ attains a non-positive minimum at a point distinct from both $0$ and $T$. More precisely, this minimum occurs around $t \simeq 0.32$.\\

Our interpretation is as follows. The trading window $[0,T]$ splits naturally into two phases. During the first phase, no information is revealed to the traders. As a result, they choose to unwind their inventories. This benefits the broker, who can match buyers and sellers and collect transaction fees $\eta$ without having to liquidate any of his own inventory. However, once the traders have unwound their inventories, the broker cannot generate profit unless he begins to reveal his information. Assuming, for example, a positive signal, the traders will then buy from the broker at cost $\eta$, while the broker can simultaneously buy on the market at a cost of $\eta^B < \eta$, thereby still making a profit.\\

Figure~\ref{fig:rep_init_0} illustrates the behavior of a representative trader starting with zero initial inventory. As expected, the trader remains inactive until information is revealed by the broker. Once this occurs, and given that the drift is positive, the trader begins to buy. He initially trades at a high rate to quickly build a sufficiently positive inventory, and then gradually reduces his trading intensity, in line with classical optimal execution models. As the time horizon approaches its end, the trader progressively liquidates part of his position to mitigate the terminal inventory penalty.\\

Conversely, during the first phase, a trader starting with a large positive inventory initially begins to sell in order to reduce his position towards zero, since the expected drift is zero at this stage. This behavior is illustrated in Figure~\ref{fig:rep_init_200}. Once the information is revealed, the trader naturally rebuilds his inventory. Such behavior is highly beneficial to the broker.

\begin{figure}[!h]
    \centering
    \includegraphics[width=0.64\textwidth]{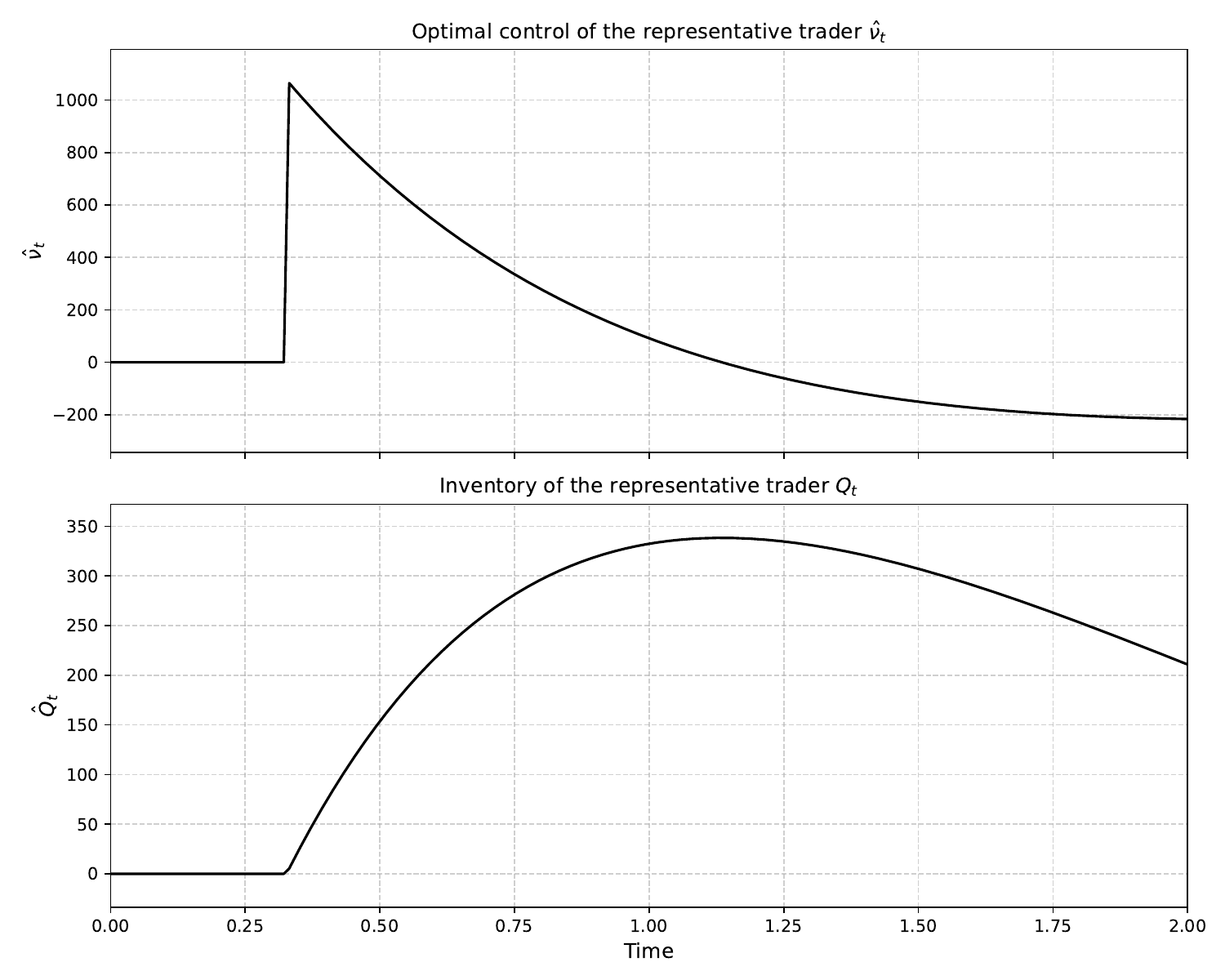}
    \caption{Behavior of a representative trader with initial inventory $Q_0=0$.}
    \label{fig:rep_init_0}
\end{figure}

\begin{figure}[!h]
    \centering
    \includegraphics[width=0.64\textwidth]{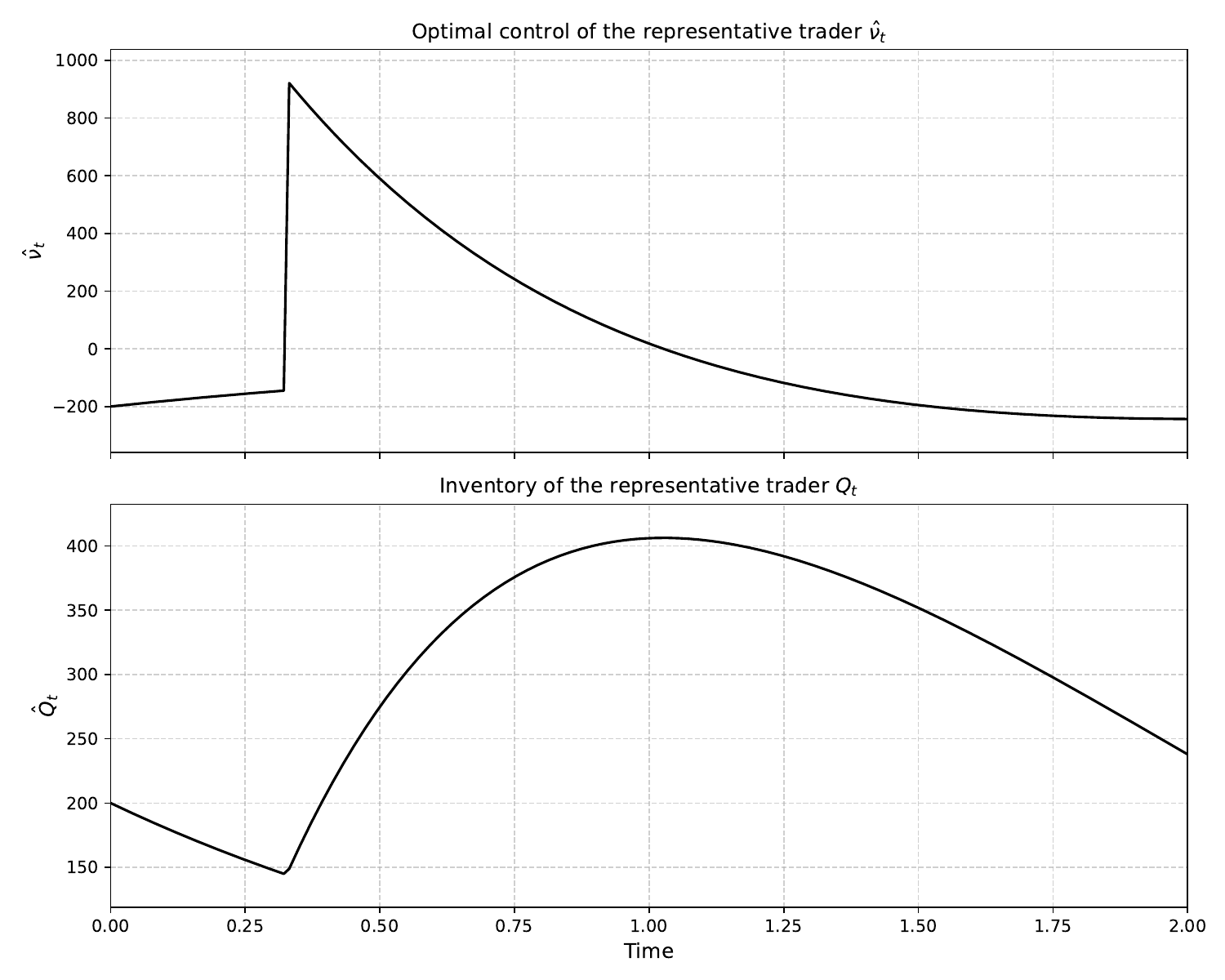}
    \caption{Behavior of a representative trader with initial inventory $Q_0=200$.}
    \label{fig:rep_init_200}
\end{figure}

Overall, Figure~\ref{fig:small_distrib} shows the distribution of small traders’ inventories at different time points, based on 10,000 simulations. We start with a Gaussian distribution for the initial inventory with zero mean and a standard deviation of 0.5. Once the information is revealed, the distribution shifts to the right, towards positive inventories, to take advantage of the drift. As the terminal time $T$ approaches, the distribution slightly moves back towards zero to account for the terminal penalty. We also observe that the standard deviation tends to decrease over time, since traders are homogeneous and target the same inventory level at any given time; the remaining differences can be explained by the variation in their initial inventory values.

\begin{figure}[!htbp]
    \centering
    \includegraphics[width=0.71\textwidth]{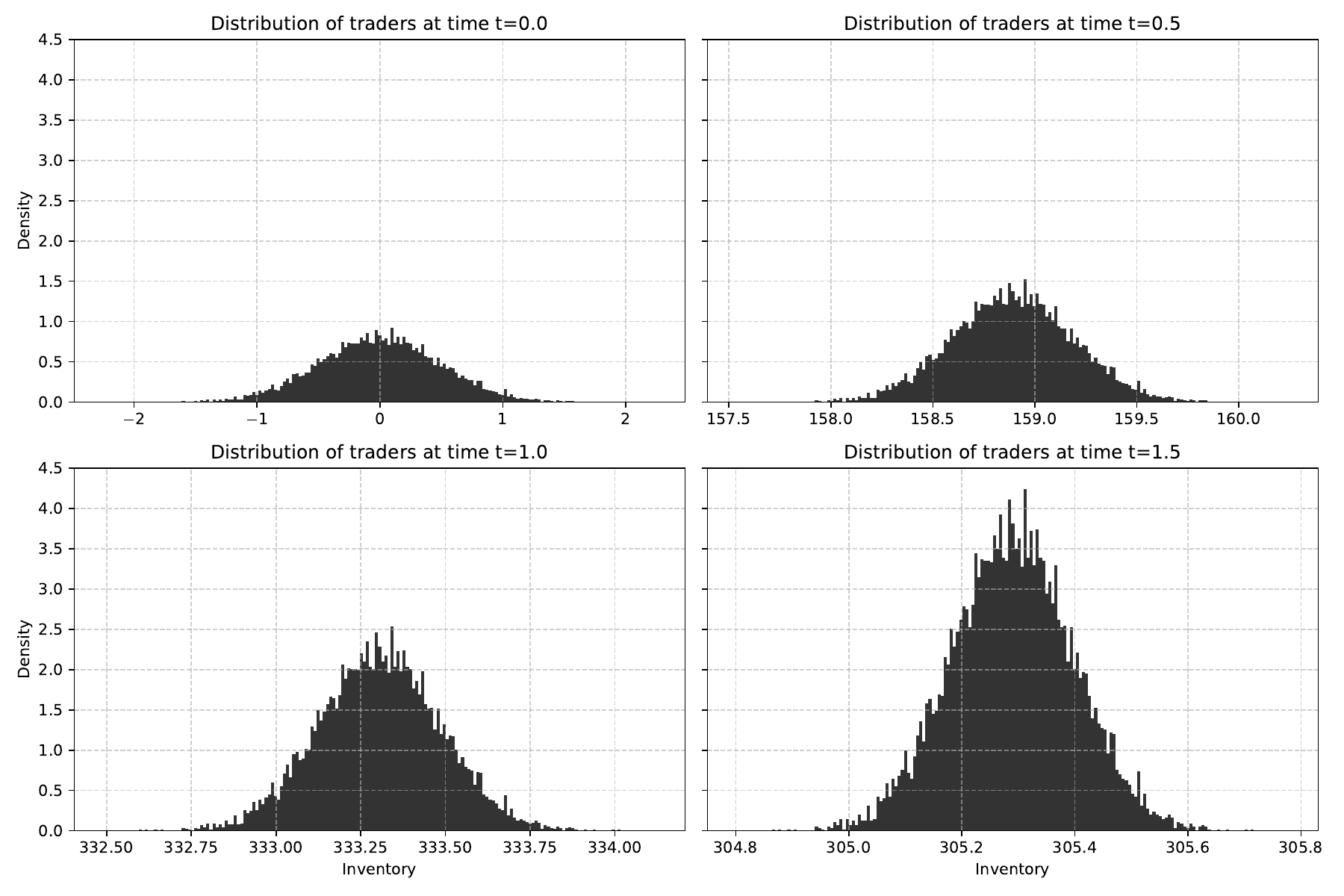}
    \caption{Distribution of the traders' inventories.}
    \label{fig:small_distrib}
\end{figure}

Finally, Figure~\ref{fig:broker} displays the behavior of the broker who starts with zero inventory. Naturally, he remains inactive until the time he reveals his information. Afterward, both his control and inventory become positive, indicating that he buys at a relatively higher rate than the traders in order to maintain a positive inventory and benefit from the drift. The terminal penalty also forces him to reduce his inventory as the horizon approaches; note that his inventory begins to decrease around $t=0.75$ while his trading rate is still slightly positive -- this occurs because this rate is no longer sufficient to offset the purchases made by the small traders.

\begin{figure}[!htbp]
    \centering
    \includegraphics[width=0.71\textwidth]{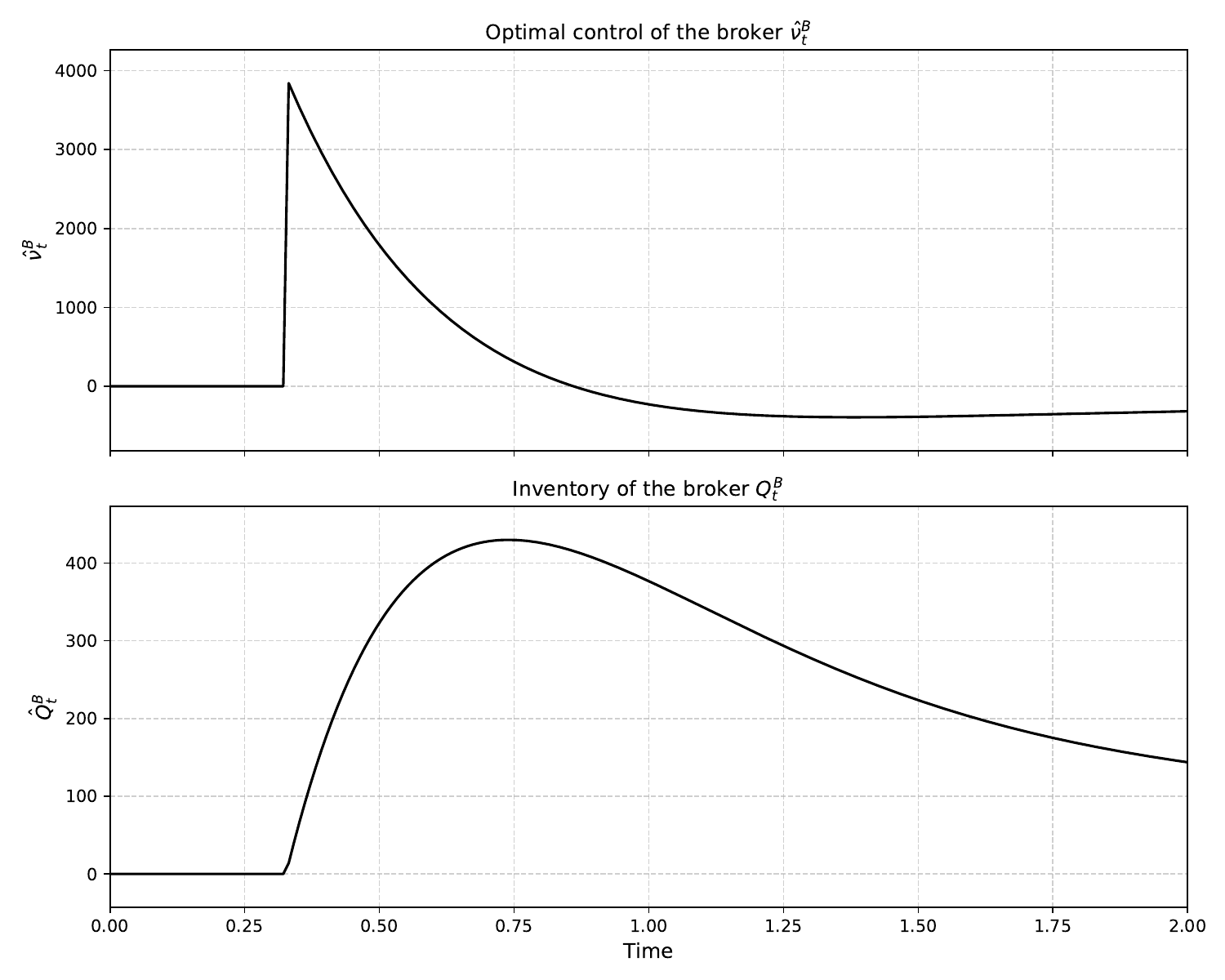}
    \caption{Behavior of the broker.}
    \label{fig:broker}
\end{figure}

\section{Conclusion}

We have studied a Stackelberg game in which an informed broker interacts with a continuum of uninformed traders, aiming to externalize inventory while concealing private information about the price drift. The model combines an open-loop control for the broker and a closed-loop mean field game for the traders, yielding explicit characterizations of optimal strategies. Our results show that the broker should initially behave as if uninformed and strategically delay the revelation of information. This delayed disclosure allows the broker to benefit from his informational advantage without inducing adverse reactions from the market too early.\\

The numerical analysis illustrates how traders adjust their inventories in response to the broker’s actions, and how the broker's strategy induces a gradual alignment of the traders' positions. These simulations confirm the theoretical insights and reveal the collective dynamics of inventory adjustment and risk management under asymmetric information.

\bibliographystyle{plainnat}
\bibliography{mfgbroker.bib}

\end{document}